\documentclass[11pt]{article}
\usepackage[margin=1in]{geometry}
\usepackage[utf8]{inputenc} %
\usepackage{booktabs}       %
\usepackage{amsfonts,yfonts}       %
\usepackage{nicefrac}       %
\usepackage{microtype}      %
\usepackage{parskip}        %
\usepackage[dvipsnames]{xcolor} %
\usepackage{amsmath, amsthm, mathtools, dsfont}
\usepackage{txfonts} %
\usepackage{color,graphicx}
\usepackage{url,hyperref}
\usepackage{enumerate}
\usepackage{algorithm}
\usepackage[noend]{algpseudocode}
\usepackage[most]{tcolorbox}
\usepackage{longfbox}
\usepackage{tikz}
\usepackage{tcolorbox}
\usepackage{framed}

\usepackage{microtype}
\usepackage{graphicx}
\usepackage{subfigure}
\usepackage{booktabs} %
\usepackage{hyperref}

\usepackage{amsmath}
\usepackage{amssymb}
\usepackage{mathtools}
\usepackage{amsthm}

\usepackage[capitalize,noabbrev]{cleveref}

\usepackage[noend]{algpseudocode}

\usepackage{tcolorbox}

\usepackage{framed}
\usepackage{mdframed}
\usepackage{mathtools, dsfont}

\theoremstyle{plain}

\newtheorem{theorem}{Theorem}

\newtheorem{lemma}[theorem]{Lemma}
\newtheorem{fact}[theorem]{Fact}

\newtheorem{claim}[theorem]{Claim}

\newtheorem{definition}[theorem]{Definition}
\newtheorem{remark}[theorem]{Remark}

\newenvironment{proofof}[1]{\begin{trivlist} \item {\bf Proof
#1:~~}}
  {\qed\end{trivlist}}

\usepackage[textsize=tiny]{todonotes}

\usepackage{tcolorbox}

\usepackage{framed}
\usepackage{mdframed}

\newcommand{\eps}{\varepsilon}

\renewcommand{\Pr}[1]{\ensuremath{\mathbf{Pr}[#1]}}
\newcommand{\ind}[1]{\ensuremath{\mathds{1}\left[#1\right]}}
\newcommand{\Ex}[1]{\ensuremath{\mathbb{E}\left[#1\right]}}

\newcommand{\Exu}[2]{\ensuremath{\mathbb{E}_{#1}\left[#2\right]}}

\DeclareMathOperator{\argmax}{arg\,max}

\newcommand{\InsertF}{{\textsc{Insert}}}
\newcommand{\DeleteF}{{\textsc{Delete}}}

\newcommand{\LevelF}{{\textsc{ConstructLevel}}}
\newcommand{\CalcSampleCountF}{{\textsc{CalcSampleCount}}}
\newcommand{\ReduceMeanF}{{\textsc{ReduceMean}}}
\newcommand{\FilterF}{{\textsc{Filter}}}

\newcommand{\mO}{O}

\newcommand{\algorithmicbreak}{\textbf{break}}
\newcommand{\Break}{\algorithmicbreak}

\newcommand{\bR}{\ensuremath{\mathbf{R}}}
\newcommand{\bQ}{\ensuremath{\mathbf{Q}}}
\newcommand{\Rp}{\ensuremath{\overline{R}}}

\newcommand{\bS}{\ensuremath{\mathbf{S}}}

\newcommand{\bM}{\ensuremath{\mathbf{m}}}
\newcommand{\bP}{\ensuremath{\mathbf{P}}}

\newcommand{\bH}{\ensuremath{\mathbf{H}}}
\newcommand{\bX}{\ensuremath{\mathbf{X}}}
\newcommand{\bRp}{\ensuremath{\mathbf{\Rp}}}

\newcommand{\bT}{\ensuremath{\mathbf{T}}}

\newcommand{\bL}{\ensuremath{\mathbf{L}}}
\newcommand{\Rc}{\ensuremath{\widehat{R}}}
\newcommand{\bRc}{\ensuremath{\mathbf{\Rc}}}
\newcommand{\epsDel}{{\epsilon_{\textnormal{del}}}}

\newcommand{\epsOpt}{{\epsilon_{\textnormal{opt}}}}
\newcommand{\epsSamp}{{\epsilon_{\textnormal{sam}}}}
\newcommand{\optSet}{{G_{\textnormal{opt}}}}
\newcommand{\opt}{{\text{OPT}}}
\newcommand{\sol}{{\text{Sol}}}
\newcommand{\pred}{{\text{pred}}}
\newcommand{\epsBuck}{{\epsilon_{\textnormal{buck}}}}

\newcommand{\rbr}[1]{\left(\,#1\,\right)}

\newcommand{\cbr}[1]{\left\{\,#1\,\right\}}

\newcommand{\ceil}[1]{\left\lceil\,#1\,\right\rceil}

\renewcommand{\Pr}[1]{\ensuremath{\mathbb{P}\left[#1\right]}}
\newcommand{\Pru}[2]{\ensuremath{\mathbb{P}_{#1}\left[#2\right]}}

\newcommand{\ReconstructF}{{\textsc{Reconstruct}}}
\newcommand{\InitF}{{\textsc{Init}}}

\newcommand{\poly}{\text{poly}}

\newcommand{\CasesIf}{{\quad \text{If} \quad }}

\newcommand{\paranth}[1]{{\left(#1 \right)}}
\newcommand{\abs}[1]{{\left|#1 \right|}}

\title{Dynamic Constrained Submodular Optimization with Polylogarithmic Update Time
\footnote{
Appears in ICML'23. 
An earlier version of this paper was submitted to the SODA'23 conference
in July 2022. In February 2023, the authors of \cite{DBLP:conf/nips/LattanziMNTZ20} contacted us to mention that one of the authors was a referee for our SODA'23 submission, they agree with a bug and they will have a fix for it in their revised arxiv paper~\cite{lattanzi2020fully}.
}
}
\author{ 
Kiarash Banihashem\thanks{Computer Science Department, University of Maryland, College Park, MD, USA. {\tt kiarash@umd.edu}.}
\and
Leyla Biabani\thanks{Department of Mathematics and Computer Science, TU Eindhoven, the Netherlands. {\tt l.biabani@tue.nl}.}
\and
Samira Goudarzi \thanks{Computer Science Department, University of Maryland, College Park, MD, USA. {\tt samirag@umd.edu}. }
\and
MohammadTaghi Hajiaghayi\thanks{Computer Science Department, University of Maryland, College Park, MD, USA. {\tt hajiagha@cs.umd.edu}.}
\and
Peyman Jabbarzade\thanks{Computer Science Department, University of Maryland, College Park, MD, USA. {\tt peymanj@cs.umd.edu}.}
\and
Morteza Monemizadeh\thanks{Department of Mathematics and Computer Science, TU Eindhoven, the Netherlands. {\tt m.monemizadeh@tue.nl}.}
\and
}

\begin{document}
\maketitle
\begin{abstract}
Maximizing a monotone submodular function under cardinality constraint $k$ is 
a core problem in machine learning and database with many basic applications, including video and data summarization, recommendation systems, feature extraction, exemplar clustering, 
and coverage problems. We study this classic problem in the fully dynamic model where 
a stream of insertions and deletions of elements of an underlying ground set is given and 
the goal is to maintain  an approximate solution using a fast update time.

A recent paper at NeurIPS'20 by Lattanzi, Mitrovic, Norouzi{-}Fard, Tarnawski, Zadimoghaddam~\cite{DBLP:conf/nips/LattanziMNTZ20} 
claims to obtain a dynamic algorithm for this problem
with a $(\frac{1}{2} -\epsilon)$ approximation ratio and
a
query complexity
bounded by
$\mathrm{poly}(\log(n),\log(k),\epsilon^{-1})$.
However, as we explain in this paper, the analysis has some important gaps.
Having a dynamic algorithm for the problem with polylogarithmic update time is even more important in light of a recent result by Chen and Peng~\cite{DBLP:conf/stoc/ChenP22} at STOC'22 
who show a matching lower bound for the problem
--
any randomized algorithm 
with a
$\frac{1}{2}+\epsilon$ 
approximation ratio
must have
an amortized query complexity 
that is
polynomial in $n$.

In this paper, we develop a simpler algorithm
for the problem that maintains a
$(\frac{1}{2}-\epsilon)$-approximate solution 
for submodular maximization under cardinality constraint $k$ using a polylogarithmic amortized update time. 
\end{abstract}

\section{Introduction}  
The problem of maximizing a monotone submodular function under cardinality constraint $k$ is defined as follows: 
Let $f: 2^V \rightarrow \mathbb{R}^{+}$ be a non-negative monotone 
submodular function defined on subsets of a ground set $V$. 
Let $k\in \mathbb{N}$ be a parameter. 
We are asked to return a set $X$ of at most $k$ elements such that $f(X)$ is maximum among all $k$ subsets of $V$. 
This problem is at the core of machine learning~\cite{DBLP:conf/nips/ElenbergDFK17,DBLP:conf/nips/Mitrovic0F0K19,DBLP:journals/spm/TohidiACGLK20}, 
data mining~\cite{DBLP:journals/arobots/WuT22,DBLP:conf/ijcai/AshkanKBW15,DBLP:journals/corr/abs-1906-11285}, and database~\cite{DBLP:conf/icml/BateniCEFMR19,DBLP:journals/neuroimage/SalehiKSSC18} with many basic applications including 
video and data summarization~\cite{DBLP:conf/nips/FeldmanK018}, recommendation systems~\cite{DBLP:conf/aaai/ParambathVC18,DBLP:conf/ijcai/AshkanKBW15,DBLP:journals/corr/abs-1906-11285,DBLP:conf/dexa/BenouaretAR19}, feature extraction~\cite{DBLP:conf/icml/BateniCEFMR19}, 
spatial search  and map exploration~\cite{DBLP:journals/arobots/WuT22}
exemplar clustering~\cite{DBLP:journals/neuroimage/SalehiKSSC18,DBLP:conf/kdd/BadanidiyuruMKK14}, sparse regression~\cite{DBLP:journals/jmlr/DasK18,DBLP:journals/sensors/TsaiT20,DBLP:journals/arobots/TsengM17} 
and coverage problems~\cite{DBLP:journals/tcs/Bar-IlanKP01,DBLP:journals/dam/Sagnol13}, to name a few.

For the submodular maximization problem under cardinality constraint $k$, the celebrated greedy algorithm due to Fisher, Nemhauser, and Wolsey~\cite{DBLP:journals/mp/NemhauserWF78} achieves  an approximation ratio of $1-1/e \approx 0.63$, which is optimal assuming $P \ne NP$.
However, this classic algorithm is inefficient when applied to \emph{modern big data}
settings, given the unique challenges of working with massive datasets. 
Motivated by these challenges, in recent years there has been a surge of interest in considering the submodular maximization problem 
under a variety of computational models such as streaming models \cite{DBLP:conf/kdd/BadanidiyuruMKK14,DBLP:conf/icml/0001MZLK19} 
and distributed models \cite{DBLP:conf/focs/BarbosaENW16,DBLP:conf/soda/LiuV19,DBLP:journals/topc/KumarMVV15,DBLP:journals/mst/McGregorV19}.

\paragraph{Related work.} 
Badanidiyuru, Mirzasoleiman, Karbasi and Krause~\cite{DBLP:conf/kdd/BadanidiyuruMKK14} 
were the first to study this problem in the insertion-only streaming model and developed a $(\frac{1}{2}-\epsilon)$-approximate 
streaming algorithm for this problem using space in $O(k\textrm{poly}(\log(n), \log(k),\frac{ 1}{\epsilon}))$. 
Recently,  Mirzasoleiman, Karbasi, and Krause \cite{DBLP:conf/icml/MirzasoleimanK017} 
studied this problem in the insertion-deletion streaming model where they obtained a $(\frac{1}{2}-\epsilon)$-approximate algorithm using $O(d^2(k\epsilon^{-1}\log k)^2)$ space and $O(dk\epsilon^{-1}\log k)$ average update time, where $d$ is an upper-bound for the number of deletions that are allowed. 
A follow-up by Kazemi, Zadimoghaddam, and Karbasi \cite{DBLP:conf/icml/0001ZK18} improved the space complexity and the average update of the first work down to $O(k\log k+d\log^2 k)$ and $O(dk\log^2 k+d\log^3 k)$.

One bottleneck of these two works is the polynomial dependency of their space and time complexities on the number of deletions 
since it takes too much time to (re)compute the solution after every insertion or deletion.  
Indeed, if the number of deletions is linear in $n = |V|$ or higher (say, $d=\Omega(n)$), it is better to rerun the offline 
algorithm~\cite{DBLP:journals/mp/NemhauserWF78} after every insertion and deletion. 
Despite their use cases, the above algorithms are unsuitable for many modern applications where data is highly \emph{dynamic}. 
For these applications such as data subset selection problem~\cite{DBLP:conf/nips/ElhamifarK17,DBLP:conf/icml/Elhamifar19}, 
movie recommendation system~\cite{DBLP:conf/wsdm/OhsakaM22,DBLP:conf/nips/00030K17}, influence maximization in social networks~\cite{DBLP:conf/kdd/ChenWY09,DBLP:journals/ton/TongWTD17,DBLP:conf/infocom/ZhangSWDT17},  
elements are continuously added and deleted, 
prompting the need for algorithms that can efficiently handle both insertions and deletions 
at the same time.

\paragraph{Submodular maximization in dynamic model.} 
Motivated by these interests, in this paper, we study  the submodular maximization problem under the cardinality constraint $k$ 
in the \emph{dynamic model}~\cite{DBLP:journals/talg/BernsteinFH21,DBLP:conf/sand/HanauerH022,DBLP:conf/soda/BhattacharyaHNW21,DBLP:journals/algorithmica/BhattacharyaCH20}. 
In this model, we are given a stream of 
insertions and deletions of elements of the underlying ground set $V$ and 
we have an oracle access to a function $f$ that returns value $f(A)$  for every subset $A \subseteq V$.
The goal is to maintain a good approximate set 
of at most $k$ elements after every insertion and deletion using a fast query complexity.

Recently, Chen and Peng~\cite{DBLP:conf/stoc/ChenP22} show 
a lower bound for submodular maximization in the dynamic model:
any randomized algorithm that achieves $(\frac{1}{2}+\epsilon)$-approximation ratio  
for dynamic submodular maximization under cardinality constraint $k$ 
requires amortized query complexity $\frac{n^{\tilde{\Omega}(\epsilon)}}{k^3}$. 
Therefore, the important question that we seek to answer is the following:

\begin{tcolorbox}[width=\linewidth, colback=white!90!gray,boxrule=0pt]
\textbf{Polylogarithmic question:}
Is there any dynamic algorithm 
for submodular maximization that maintains a $(\frac{1}{2}-\epsilon)$-approximate solution under cardinality constraint $k$ 
with a query complexity $\textrm{poly}(\log(n),\log(k),1/\epsilon)$  
where $n$ is the size of underlying ground set $V$? 
\end{tcolorbox}

Very recently,  Lattanzi, Mitrovic, Norouzi-Fard, Tarnawski, and Zadimoghaddam~\cite{DBLP:conf/nips/LattanziMNTZ20} 
and Monemizadeh~\cite{DBLP:conf/nips/Monemizadeh20} tried to answer this question. 
The first worked claimed to develop a dynamic algorithm 
that maintains an expected amortized $(\frac{1}{2}-\epsilon)$-approximate solution for this problem 
using an amortized expected query complexity of $\textrm{poly}(\log(n),\log(k),1/\epsilon)$. 
The second paper proposed a randomized dynamic $(\frac{1}{2}-\epsilon)$-approximation algorithm with 
expected amortized query complexity of $O(k^2 \textrm{poly}(\log(n),1/\epsilon))$, 
so it partially answers our question, but not fully. 
Therefore, only the first paper~\cite{DBLP:conf/nips/LattanziMNTZ20}  claims to answer our question. 
However, as we explain later in this paper, the proof of this dynamic algorithm has important gaps in the analysis.

In this paper, we answer the polylogarithmic question affirmatively;
we develop a new simpler algorithm 
which maintains a $(\frac{1}{2}-\epsilon)$-approximate solution 
for submodular maximization problem under cardinality constraint $k$ with a poly-logarithmic amortized query complexity.

\subsection{Preliminaries}
\textbf{Submodular functions.}
Given a finite set of elements $V$,
a function $f:2^{V} \to \mathbb{R}$ is \emph{submodular} if
\begin{align*}
    f(A \cup B) + f(A \cap B) \le f(A) + f(B).
\end{align*}
for all $A, B \subseteq V$.
In addition, function $f$ is called \emph{monotone} if
$f(A) \le f(B)$ for all $A \subseteq B$, and
is called \emph{nonnegative} if
$f(A) \ge 0$ for all $A \subseteq V$. 
In this paper, we consider non-negative monotone submodular functions. In addition, we assume that
$f(\emptyset) = 0$.\footnote{This assumption is without loss of generality since if $f(\emptyset) > 0$, then
defining the function $g$ as $g(V) := f(V) - f(\emptyset)$, $g$ will be a monotone submodular function as well.
In addition, optimizing $f$ is equivalent to optimizing $g$, and any algorithm obtaining approximation ratio $\alpha < 1$ for $g$, would obtain approximation ratio at least $\alpha$ for $f$ as well.}

\textbf{Oracle access.}
In this paper, we assume that our algorithm has access to an
\emph{oracle} that outputs value $f(A)$ when it is queried an arbitrary set $A$. 
We measure the \emph{time complexity} of a dynamic algorithm 
in terms of its \emph{query complexity} which is the number of queries it makes to the oracle.

\textbf{Dynamic model.}
We consider the classical dynamic model where we are given a sequence of insertions and deletions of elements, 
and the goal is to maintain an $\alpha$-approximate solution of size at most $k$ at any time $t$.
We assume without loss of generality that once 
an element
is deleted, 
it is never re-inserted. Indeed, the reinsertion of an element $e$ can always be simulated by inserting an identical copy of $e$.

We refer to the sequence of insertions and deletions as the \emph{update stream}, and assume that it is chosen by a non-adaptive adversary. The adversary is assumed to have knowledge of our algorithm, and can control the sequence of insertions and deletions, but does not have access to the random bits our algorithm uses.

\textbf{Notation.} Given integers $a, b$, we use the notation $[a, b]$ to denote the set $\{a, \dots, b\}$ and $[a]$ to denote $[1, a]$. We use $\ind{.}$ to denote the \emph{indicator function}, i.e., $\ind{A}$ equals one if $A$ is true and equals zero otherwise.
~\\
We use $\Pr{.}$ and $\Ex{.}$ to denote the probability and expectation, respectively. For an event $A$ satisfying $\Pr{A} > 0$, we use $\Pr{. | A}$ and $\Ex{. | A}$ to denote the \emph{conditional probability} and \emph{conditional expectation}.
In this paper, we frequently condition on the value of a random variable. 
To avoid confusion, we will often use bold letters to denote random variables and non-bold letters to denote their values, e.g., $\Ex{\mathbf{X} | \mathbf{Y} =Y}$ denotes the expectation of the random variable $\mathbf{X}$, conditioned on the random variable $\mathbf{Y}$ attaining the value $Y$.
~\\
Given a submodular function $f$ and sets $A$ and $B$,
we will use $f(A|B)$ to denote
the value $f(A \cup B) - f(A)$. 

\section{Polylogarithmic algorithm}
\subsection{Prior work}
\label{sec:existing_problems}
In this section, we highlight three main problems with the analysis of the algorithm in~\cite{DBLP:conf/nips/LattanziMNTZ20}.

On a high level, the algorithm obtains
a solution \sol{} as union of sets $S_{i, j}$ such that
each $S_{i, j}$ is claimed to satisfy
$\Ex{f(S_{i, j} | S_{\pred(i, j)})} \ge (1-\epsilon) \tau |S_{i, j}|$, where
$\tau=\frac{\opt}{2k}$, the value $OPT$ is the optimal value of any subset of ground set $V$ of size at most $k$, and 
$S_{\pred(i, j)}$ denotes the set of elements
sampled before $S_{i, j}$. The algorithm stops when it has sampled $k$ elements,
or when it is no longer possible to sample any element with the desired property.

Next, we explain three issues with the algorithm's analysis.

\paragraph{First Issue: Incorrect conditioning in proof.}
The proof of Theorem 5.1 (that proves the $(1/2 - \epsilon)$-approximation guarantee) 
does not consider the effect of conditional probability when bounding the submodular value of the samples.
Specifically, to prove that the approximation factor of the reported  set is $\frac{1}{2} - \epsilon$,
the proof considers two cases. 
The first case is when there are $k$ elements sampled.
In this case, it is claimed that the submodular value of output set $\sol$ satisfies $f(\sol) \ge \frac{\opt}{2} - \epsilon$
since, in expectation, each $S_{i, j}$ contributed $(1-\epsilon) \tau |S_{i, j}|$.
Given the assumption $|\sol| = k$ however, one
needs to
analyze the \emph{conditional expectation} of $f(S_{i, j} | S_{\pred(i, j)})$. In other words,
just because $\Ex{f(S_{i, j} | S_{\pred(i, j)})} \ge (1-\epsilon) \tau |S_{i, j}|$, 
does not mean that 
$\Ex{f(S_{i, j} | S_{\pred(i, j)}) \Big| |\sol| = k} \ge (1-\epsilon) \tau |S_{i, j}|$.

Intuitively, one can expect that whenever samples $S_{i, j}$ have high quality,
i.e., $f(S_{i, j} | S_{\pred(i, j)})$ is high, the algorithm would be more likely to terminate with less than $k$ samples since the obtained samples already have high quality and
the remaining elements may contribute little to them. 
This means that the condition $|\sol| = k$ may introduce a negative bias on the value of
$f(S_{i, j} | S_{\pred(i, j)})$.%

\paragraph{Second Issue: Lack of analysis for biased samples.}
Another 
issue is 
that it is not clear whether the identity
\begin{equation}
\Ex{f(S_{i, j} | S_{\pred(i, j)})} \ge (1-\epsilon)  \tau |S_{i, j}|
\label{eq:mywuer}
\end{equation}
holds in the early step of the analysis, even without the extra condition that $|\sol| = k$.
The analysis
considers the value of
$S_{i, j}$ and $S_{\pred(i, j)}$
the last time $\textsc{Level-Construct}(\ell)$
was called. It is then claimed
that because $S_{i, j}$ was obtained
by invoking the peeling algorithm,
\eqref{eq:mywuer} should hold. 
Importantly however, the analysis of the peeling
algorithm assumes that the chosen elements
we obtained by sampling
\emph{uniformly at random}.
While the claim may hold if the values
of $S_{\pred(i, j)}$ and $A_{i, j}$ (the set that $S_{i, j}$ is sampled from) were fixed, here
the values $S_{\pred(i, j)}$ can only be calculated randomly, by looking back at the last time
$\textsc{Level-Construct}(\ell)$ is called. Since this time itself depends on the value of $S_{i, j}$, this looking back process may introduce bias, and it is not clear how the bias should be handled.

To illustrate the issue, consider the following simple decremental process. Assume that we are given a set $V$ and we sample $k$ elements from $V$ to obtain a set $S$. An adversary then deletes an element $e$ from $V$ and,
if an element in $S$ is deleted, we take $k$ fresh samples from $V$.
Let $(V_1, S_1)$ and $(V_2, S_2)$ denote the value of $(V, S)$ before and after the deletion respectively and let $(V', S')$ denote the value of $(V, S)$ the last time we sampled $k$ elements. 
At first glance,
it may seem that the elements of $S'$ are $k$ random elements from $V'$, this is not the case. Indeed, if $V'=V_1$, then we know that $k$ fresh samples were \emph{not taken}, which implies $e\notin S'$. Since $k$ samples taken uniformly at random would contain $e$ with positive probability, this shows that the samples are not uniformly at random. 
\footnote{
  The bias here is reminiscent of the so-called random incidence paradox,
  also known as the waiting time paradox and the related inspection paradox \cite{ross2003inspection}.
  Consider a Poisson process that has been running forever and assume that
  the arrival rate is $\lambda=4$ per hour.
  This means that the average length between two successive arrivals
  is $\frac{1}{\lambda}$.
  One possible way to estimate this arrival
  rate is to pick an arbitrary point $x$, e.g., $x=0$, and measure the length of time between the
  first arrival after $x$ and the last arrival before $x$.
  It can be shown however that this results in an estimator with expectation $\frac{2}{\lambda}$
  and is therefore biased.
  To understand why, consider the segmentation of the real line caused by arrivals,
  i.e., each segment represents the interval between two successive arrivals.
  Intuitively, since we are choosing $x$ as a random point in the real line,
  $x$ is more likely to fall in the longer intervals than the short ones. This results in a bias
  in favor of the longer intervals in the estimator, leading to a larger estimate of $\frac{2}{\lambda}$. 
  }

The algorithm of \cite{DBLP:conf/nips/LattanziMNTZ20} has a ``resampling'' condition simliar to what is described above.
In each level $\ell$, they maintain a series of buckets,
where the samples $S_{i, \ell}$ on each bucket are obtained by repeatedly invoking the peeling algorithm and depend on the samples of the previous samples $S_{j, \ell}$ for $j< i$. 
Once an $\epsilon$-fraction of $S_{i, \ell}$ is deleted, the entire level $\ell$ is reconstructed.

\textbf{Third Issue: Expectation bound.
\footnote{We thank an anonymous reviewer in SODA'23 for pointing out this issue.}
}
Lemma $C.4$ 
bounds the expectation of the ratio of two expressions, 
but the applications of this lemma assumes that the bound applies to the ratios of the expectations, which is not the same thing.

Given the above issues, a natural question is whether it is even possible
to obtain $\frac{1}{2} - \epsilon$ approximation factor with polylogarithmic query complexity. 
Indeed, as we explained in above, Chen and Peng~\cite{DBLP:conf/stoc/ChenP22} show 
that obtaining $\frac{1}{2} + \epsilon$ approximation requires
an query complexity that is polynomial in $n$ (in fact, requires amortized query complexity $\frac{n^{\tilde{\Omega}(\epsilon)}}{k^3}$),
suggesting that $\frac{1}{2} -\epsilon$ may be impossible as well.
In this paper, we show that this is not the case.
We develop a  simpler algorithm 
than that proposed in~\cite{DBLP:conf/nips/LattanziMNTZ20} 
which maintains an $(\frac{1}{2}-\epsilon)$-approximate solution 
for submodular maximization under cardinality constraint $k$ with a poly-logarithmic query complexity. 
We emphasize that while our algorithm is simpler,
its requires careful analysis to avoid the aforementioned issues.

\subsection{Overview of our algorithm}
\paragraph{Offline algorithm.} 
In this section, we first present the offline version of our algorithm. 
Later, we show how to support insertions and deletions.
We first remove all elements that have submodular value less than $\tau = OPT/2k$ and 
let $R_1$ be the remaining elements. 
Let $G_0 = \emptyset$ and $i=1$. In each iteration, we
bucket
the elements of $R_i$ in 
based on  their relative
marginal gain to $G_{i-1}$, i.e., $f(e | G_{i-1})$ for $e \in R_i$, such that
all the elements in the same bucket have the same value of $f(e |G_{i-1})$
up to a factor of $1+\epsBuck$.
We use $R_i^{(b(i))}$ to denote
the largest (maximum size) bucket.

Next, for a suitable number $m_i$,
we take a uniformly random subset of size $m_i$ from the largest bucket and we denote it by $S_i$.
We then add $S_i$ to $G_{i-1}$ to form $G_{i}$. 
Next, we remove all elements
$f(e | G_{i}) < \tau$ from $R_{i}$ to form $R_{i+1}$ using the function: 
    \[
    	\begin{split}
		&R_{i+1} = \FilterF{}(R_i, G_i, \tau) = 
        \begin{cases}
			    \{e \in R_i: f(e | G_i) \ge \tau\}, & \text{if $|G_i| < k$}.\\
			    \emptyset, & \text{otherwise}.
	  	\end{cases}
	\end{split}
   \]

Later, we repeat the above process by setting $i = i+1$.
The process is continued until there is no element with $f(e | G_{i}) \ge \tau$ left,
or we have chosen $k$ elements.
The final value of $G_{i}$, denoted by $G_{T}$, is reported as the output of the algorithm.
After our offline algorithm has executed, sets $R_{0}, \dots, R_{T+1}, S_{1}, \dots, S_{T}, G_{0}, \dots G_{T}$ will
satisfy the following properties

\begin{tcolorbox}[width=\linewidth, colback=white!90!gray,boxrule=0pt]
\textbf{Properties:} \\
(1) $R_i = \FilterF{}(R_{i-1}, G_{i-1}, \tau)$ for $i\in [1, T+1]$.\\
(2) Each $S_i$ is a uniformly random subset of size $m_i$ from
    the largest bucket of $R_i$, denoted by $R_i^{(b(i))}$ and
    $G_{i} = \bigcup_{j\le i}S_j$.
\end{tcolorbox}

\paragraph{Dynamic algorithm.}
To handle insertions and deletions, we maintain modifications of the above properties after the updates in the stream.
These properties
will ensure that our algorithm has an approximation factor of $\frac{1}{2} - \epsilon$ using 
expected amortized query complexity that is polylogarithmic in $n, k$.

To handle insertions,
for each level $i$,
we maintain a set $\Rp_i \supseteq R_i$ that will temporarily
hold the elements in a level before they are processed.
More formally, each time an element $v$ is inserted,
we add $v$ to all $\Rp_i$ such that
$f(v | G_{i-1}) \ge \tau$.
Once the size of $\Rp_i$ is at least $\frac{3}{2}|R_i|$,
we reconstruct the level $i$.
A formal pseudocode is shown in Algorithm~\ref{alg:insert}.
~\\
To handle deletions,
we keep track of the deleted elements
by adding them to a set $D$.
For each level $i$, we keep track of the number of elements
deleted in the bucket $R_i^{(b(i))}$ that we had chosen
to sample $S_i$ from. Once an $\epsilon$-fraction of
these elements is deleted, we restart the offline algorithm from level $i$
by invoking \ReconstructF{}(i).
A formal pseudocode is provided in Algorithm \ref{alg:delete}.

We note that our construction
is different from
that of \cite{DBLP:conf/nips/LattanziMNTZ20} in a number of ways.
Firstly, their algorithm maintains
multiple buckets in each level and each bucket contains its own set of samples $S_{i, \ell}$,
where the set $S_{i, \ell}$ is sampled after $S_{i-1, \ell}$.
Once an $\epsilon$-fraction of $S_{j, \ell}$ is deleted for some $j$, the entire bucket $\ell$ is reconstructed.
Therefore, the resampling condition for $S_{i, \ell}$ depends on all $S_{j, \ell}$, including
$S_{i, \ell}$ itself and $S_{j, \ell}$ for $j> \ell$ which in turn depend on $S_{i, \ell}$. 
In contrast, our reconstruction condition for $S_{i}$ depends only on $S_{1}, \dots, S_{i-1}$ and does not depend on $S_{i}$ or any $S_{j}$ for $j > i$.
As we show in Section \ref{sec:invariant} (Lemma \ref{lm:uniform_lazy}),
this choice allows us
to prove that
the samples $S_i$ are \emph{always} a uniformly random subset of $R_i^{(b(i))}$.

Now, we explain how we choose sample size $m_i$ in each level. Our choice here is based on the peeling/threshold-sampling algorithm \cite{fahrbach2019submodular}, though we require new analysis in the proofs.
One possibility is to set $m_i=1$. This ensures that
$f(S_i | G_{i-1}) \ge |S_i| \cdot \tau$ since all the elements in $R_i$ are guaranteed
to satisfy $f(e | G_{i-1}) \ge \tau$.
The problem with this approach is that the number of levels $T$ can be polynomial in $k$
as we may need to sample every element of a solution set which is of size $k$ 
in a separate level. This would in turn lead to a polynomial query complexity.

Another extreme is to choose a very large integer, e.g., set $m_i=k$.
The problem with this approach however is that the sample $S_i$
may have low relative marginal gain to $G_{i-1}$, i.e.,
$f(S_{i} | G_{i-1})$ may be low. Indeed,
once we sample one element
$s_i$ from $R_i^{(b(i))}$, 
we expect the marginal gain of the remaining element
to drop since $f$ is submodular.

A good value of $m_i$ needs to balance the above trade off.
Let $\tau^{(i)}$ denote the minimum threshold corresponding to $R_i^{(b(i))}$, i.e.,
$e \in [\tau^{(i)}, (1+\epsilon)\tau^{(i)})$ for $e\in R_i^{(b(i))}$. 
Intuitively, we want to choose an $m_i$ such that
\begin{enumerate}
  \item In expectation, each element in $S_i$ adds at least
    $(1-\epsilon) \cdot \tau^{(i)}$ to the value of the output, i.e., 
    $f(S_{i} | G_{i-1}) \ge (1-\epsilon) \cdot |S_i| \cdot \tau^{(i)}$.
  \item The value $m_i$ is large enough to ensure that $|R_{i+1}|$
    is considerably smaller than $|R_i|$, ensuring that the number of levels is not too large.
\end{enumerate}
The first property is important for insuring the approximation guarantee of our algorithm,
while the second property controls its query complexity.

The main idea is to choose the largest $m_i$ that satisfies a modification of the first property.
Given a parameter $\epsilon$, we search for the largest integer $m_i$ such that
when sampling $m_i$ elements at random, the \emph{last element} has marginal gain at least
$\tau^{(i)}$ with probability at least $1-\epsilon$. 
Since the last element is likely to be the worst element because of submodularity, this ensures the first property.
In addition, since we chosen the largest integer $m_i$, a random sample
of the remaining elements should contribute $\tau^{(i)}$ with probability at most $1-\epsilon$ since
it is effectively the last elements of a sample of size $m_i + 1$.
This ensures that an $\epsilon$ fraction of the elements in $R_i^{(b(i))}$ will be removed
from the $b(i)$-th bucket in the next level, which in turn allows us to bound the number of levels.

To find the largest such $m_i$, we binary search
over the set of all possible values $m'$ for $m_i$.
For each $m'$, we test whether satisfies the first property above by sampling $S'$ for $\mO(\frac{1}{\epsilon^2} \cdot \log(k/\epsilon))$ trials, and testing whether its last element has marginal gain at least $\tau$ in more than $(1-\epsilon)$ fraction of the trials. 
A standard Chernoff bound then shows that $m_i$ satisfies the mentioned properties with high probability.

Finally, we relax the assumption of known \opt{} by maintaining multiple parallel
runs, indexed by $p \in \mathbb{Z}$.
For each run $p$, we will use the value
$\opt_p = (1+\epsOpt)^p$ for $\opt$ in the algorithm
and only insert elements with $f(e) \in [\epsilon \cdot \frac{\opt_p}{2k}, \opt_p]$.
We always output the set with maximum value of $f$ across all the runs.
This increases the query complexity of our algorithm by a factor of at most
$O(\log_{1+\epsOpt}(k/\epsilon))$
, since each element needs to be inserted into $\lceil\log_{1+\epsOpt} (2k/\epsilon)\rceil$ runs, 
and reduces the approximation guarantee by at most
$\epsilon$, since the discarded elements affect the solution by at most $\epsilon \cdot \opt$.

\subsection{Overview of techniques}
\paragraph{Approximation factor}
We start by giving an overview of our proof for the approximation factor of the algorithm.
As we show in Section \ref{sec:invariant} (Lemmas \ref{lm:invariant_filter} and \ref{lm:uniform_lazy}), the output of our algorithm satisfies the following important properties.

\begin{enumerate}
    \item 
      For all $i \in [0, T]$,
      defining $\Rc_i := \Rp_i \backslash D$,
      \begin{math}
        \Rc_{i + 1} = \FilterF{}(\Rc_i, G_i).
      \end{math}
      In addition, $\Rc_{T+1}=\emptyset$.
    \item 
      Conditioned on the values $S_1, \dots S_{i-1}$ and $m_{i}$, the set $S_i$ is a uniform subset of size $m_i$ from $R_i^{(b(i))}$. In other words,
      \[
      	\begin{split}
        &\Pr{\bS_i = S | \bT \ge i, \bS_1=S_1, \dots, \bS_{i-1}=S_{i-1}, \bM_i=m_i)}
        = \frac{1}{\binom{|R_i^{(b(i))}|}{m_i}}\ind{S \subseteq R_i^{(b(i))} \land |S| = m_i} \enspace ,	
	\end{split}
      \]
      where we have used bold letters to denote random variables and non-bold letters to denote their values.
\end{enumerate}

As mentioned earlier, the integers $m_i$ are chosen such that  in expectation, each element of $S_i$ contributes at least $(1-\epsilon)\tau$, i.e., 
$\Ex{f(S_i | G_{i-1})} \ge (1-\epsilon) \cdot |S_i| \cdot\tau$.
If the sample quality was deterministic, i.e,
$f(S_i | G_{i-1})$ was guaranteed to always be at least
$(1-\epsilon) \tau$, then we could have used a well-known  argument
that considers two separate
cases based on whether or not $G_T$ has $k$ elements.
If it has $k$ elements, then since each sample added, on average,
$(1-\epsilon)\tau$ to final $f$, then the claim would hold.
If it does not have $k$ elements, then since
$\Rc_{i + 1} = \FilterF{}(\Rc_i, G_i)$ and
$\Rc_{T+1} = \emptyset$, 
each element $e$ in the optimal set
must satisfy $f(e | G_T) < \tau$, which can then be used to show that
$f(G_T) \ge f(\optSet \cup G) - k\tau \ge \opt - k\tau = \frac{\opt}{2}$. 

However, the issue with this type of analysis is since conditioning on the event
$|G_T|=k$, would mean that we can no longer guarantee
$\Ex{f(S_i | G_{i-1})} \ge (1-\epsilon)\tau|S_i|$.
In other words, just because
$\Ex{f(S_i | G_{i-1}) | \bT \ge i} \ge (1-\epsilon)\tau|S_i|$ holds,
it does not mean that
$\Ex{f(S_i | G_{i-1}) \big| \bT \ge i, |G_T| =k} \ge (1-\epsilon)\tau|S_i|$ would hold as well.
Indeed, one can expect that if sample set $S_i$ has a high quality, i.e, 
$f(S_i |G_{i-1})$ is large, then the algorithm would be more likely to reach
$OPT/2$ with less than $k$ elements. 
The condition $|G_T| \le k$ therefore induces a negative bias on the quality of the samples.
Further complicating the analysis is the fact that some of the samples are deleted
and the output of our algorithm is $G_T \backslash D$, not $G_T$.

To address this issue, we introduce a novel relaxation function that may be of independent interest.
By carefully unifying the cases of $|G_T|=k$ and $|G_T| < k$,
we first show that our new relaxation function is a lower bound on $f(G_T \backslash D)$.
We then bound the contribution of each level to the relaxation function
to prove the desired $\frac{\opt}{2}-\epsilon$ bound.
We refer the reader to Section \ref{sec:proof_approx} for more details.

\paragraph{Query complexity.}
To bound the query complexity of our algorithm, we show that our choice of $m_i$ means the number of levels is polylogarithmic in $n, k$. 
The choice ensures that at least $\epsilon \cdot |R_i^{(b(i))}|$ elements in $R_i^{(b(i))}$ will either be moved to a \emph{lower} bucket in $R_{i+1}$, or they will not appear in $R_{i+1}$ altogether.
Since the number of bueckets is bounded by $\log_{1+\epsBuck}(\frac{\opt}{\tau}) = \log_{1+\epsBuck}(2k)$, each element can only go down at most $\log_{1+\epsBuck}(2k)$ times.
Our proof formalizes this intuition, though there are some subtleties given that the above guarantees hold only in expectation.

The polylogarithmic bound on the number of levels implies that
each call to \ReconstructF{}(i) makes at most
$\mO(|R_i| \cdot \poly(\log(n), \log(k), \frac{1}{\epsilon}))$ queries.
Given the choice of reconstruction conditions in our algorithm, when an insertion 
or deletion triggers a call to \ReconstructF{}(i), we can charge it back to at least $\frac{|R_i|}{\poly(log(k), \frac{1}{\epsilon})}$ elements that triggered it. 
This implies that each $\ReconstructF{}(i)$ charges back a polylogarithmic number of queries to each element. 
In addition, each update
can only be charged once by each level (the first time the level is reconstructed after this update).
Since the number of levels is polylogarithmic, this implies a polylogarithmic bound on the query complexity. 
We refer the reader to Section \ref{sec:query} for a more detailed analysis.

\textbf{Comparison with prior work.}
We briefly highlight the main
differences in the algorithm and analysis that allow
us to sidestep the aforementioned issues
of the analysis in~\cite{DBLP:conf/nips/LattanziMNTZ20}. 

The first main difference is the design of multi-level structure and reconstruction condition. Our structure only samples elements once per level, and our reconstruction condition considers the effect of deletions on the set of elements these samples were chosen from, not the actual elements themselves. This ensures that whether or not a level is reconstructed is independent of its samples. 
This effectively resolves the biasing issue discussed in Section \ref{sec:existing_problems}.
A consequence of these changes, however, is that the the existing analysis needs to be altered significantly for both the approximation guarantee and query complexity.

The second main difference is the introduction of a relaxation function that always lower bounds the value of the output $f(G_T \backslash D)$. 
The relaxation function unifies the two cases of $|G_T| =k$ and
$|G_T| < k$, allowing us to sidestep the conditioning issue discussed in Section \ref{sec:existing_problems}.

\section{Proposed algorithm}
\label{sec:lazy}
In this section, we present
our dynamic algorithm for submodular maximization under cardinality constraint $k$ using 
polylogarithmic number of oracle queries. 
We start by presenting an offline algorithm for the problem
in Section \ref{sec:alg_level}. 
We then show how to extend this to a dynamic setting by considering lazy updates
in Sections \ref{sec:alg_insert}, and \ref{sec:alg_delete}.
In these sections, we assume access to a parameter $\opt$, which
estimates the value of the optimal solution up to a factor of $1+\epsilon$. More formally, we assume that $f(\optSet) \le \opt{} \le (1+\epsOpt)f(\optSet{})$ where $\optSet$ denotes the optimal solution.
In section \ref{sec:alg_parallel}, we show how to remove this restriction by considering parallel runs.

\subsection{Offline algorithm}
\label{sec:alg_level}
We now present an offline algorithm for the problem.
In this section, we assume that the value of the optimal solution
which we denote by \opt{} is known to the algorithm. As we will show in our analysis,
\opt{} does not need to be an exact estimate and our results still hold
as long as it is approximately correct. Later in Section \ref{sec:alg_parallel},
we will remove this restriction all together.

Setting the threshold $\tau$ to be $\frac{\opt}{2k}$,
our algorithm starts by removing all elements that have submodular value less than $\tau$,
collecting the remaining elements in the set $R_1$. 
Defining the set $G_0$ as $\emptyset$ and starting with $i := 1$,
in each iteration we group the samples in buckets based on their relative
marginal gain to $G_{i-1}$, i.e., $f(e | G_{i-1})$, such that
all the elements in the same bucket have the same value of $f(e |G_{i-1})$
up to a factor of $1+\epsBuck$.
In other words, all the elements $e$ in the same bucket
satisfy $f(e | G_{i-1}) \in [\tau^{(i)}, (1+\epsBuck) \cdot \tau^{(i)}) $
for some threshold $\tau^{(i)}$.
This requires
$\log_{1+\epsBuck}(2k)$ buckets as 
\begin{enumerate}
  \item We assume that $f(e|G_{i-1}) \ge \tau$ for all $e \in R_{i}$. 
    This was the case for $i=1$, and we will also ensure it to be the case for $i\ge1$
    by removing all elements with $f(e | G_{i-1}) < \tau$ from $R_i$.
  \item All elements have $f(e|G_{i-1}) \le f(e) \le \opt = 2k\cdot\tau$.
\end{enumerate}
Next, we take a uniformly random subset of $m_i$ from the largest bucket
for a suitable number $m_i$, forming the samples $S_i$. 
We will explain how to choose $m_i$ in Section \ref{sec:alg_sample}.
We then add $S_i$ to $G_{i-1}$ to form $G_{i}$, remove all elements
$f(e | G_{i}) \le \tau$ from $R_{i}$ to form $R_{i+1}$ and repeat the above process
with $i\gets i+1$.
The process is continued until there are no elements with $f(e | G_{i}) \ge \tau$ left,
or we have chosen $k$ elements.
The final value of $G_{i}$, denoted by $G_{T}$, is given as the algorithm's output.
A formal pseudocode is provided in Algorithm \ref{alg:offline}.
The values $\Rp$ and $D$ in the algorithm will be necessary for the next section
and for now, we can think of them as $\Rp_i = R_i$ and $D = \emptyset$. 
As seen in Algorithm \ref{alg:offline},
the main part of our design, which is the function \ReconstructF{},
is more general than the description above and can, effectively,
start midway from a set $R_{i}$ for $i \ge 1$.
This property will play a crucial role in the upcoming sections as we frequently require reconstructing a portion of our data structure to handle insertions and deletions.

\begin{algorithm}[ht]
  \caption{Offline algorithm}
  \label{alg:offline}
  \begin{algorithmic}[1]
    \Procedure{\InitF{}}{$V, \opt$}
      \State $R_0 \gets V$, \quad $\tau \gets \frac{\opt}{2k}$, \quad $G_0 \gets \emptyset$
      \State $D \gets \emptyset$, \quad $\Rp_0 \gets R_0$
      \State $R_1 \gets \FilterF{}(R_0, G_0, \tau)$, \quad $\Rp_1 \gets R_1$
      \State $\ReconstructF{}(1)$
    \EndProcedure
    \Procedure{\ReconstructF}{$i$}
      \State $R_i \gets \Rp_i \backslash D$, \quad $\Rp_i \gets R_i$
      \label{line:level_first_line}
      \While{$R_i \ne \emptyset$}\label{line:level_break}
        \For{$j \in [0, \left\lfloor  \log_{1+\epsBuck} (2k)\right\rfloor]$}
          \State $R_i^{(j)} \gets \{e \in R_i: \frac{f(e | G_{i-1})}{\tau} \in [(1+\epsBuck)^{j}, (1+\epsBuck)^{j+1})\}$
        \EndFor
        \State $b(i) \in \argmax_{j} \abs{R_i^{(j)}}$ 
        \State $\tau^{(i)} \gets (1+\epsBuck)^{b(i)} \cdot \tau$ 
        \State $m_i \gets \CalcSampleCountF{}(R_i^{(b(i))}, G_{i-1}, \tau^{(i)})$
        \State $S_i \gets$ Uniform subset of size $m_i$ from $R_i^{(b(i))}$ \label{line:level_choose_sample}
        \State $G_{i} \gets G_{i-1} \cup S_i$
        \State $R_{i+1} \gets \FilterF{}(R_i, G_i, \tau)$, \quad $\Rp_{i+1} \gets R_{i+1}$
        \label{line:level_filter}
        \State $i \gets i + 1$
      \EndWhile
      \State $T \gets i - 1$
    \EndProcedure
    \Function{\FilterF}{$R', G', \tau'$}
      \If{$|G'| = k$}
        \State \Return $\emptyset$
      \Else
        \State \Return $\{e \in R': f(e | G') \ge \tau'\}$
      \EndIf
    \EndFunction
  \end{algorithmic}
\end{algorithm}

Next, we show how to handle insertions and deletions in our algorithm.
\subsection{Insertion}
\label{sec:alg_insert}
We take a lazy approach for dealing
with insertions;
for each level $i$,
we maintain a set $\Rp_i \supseteq R_i$ that will temporarily
hold the elements in a level, and process these elements once the number of elements is sufficiently large.
More formally, each time an element $v$ is inserted,
we add $v$ to all $\Rp_i$ such that
$f(v | G_{i-1}) \ge \tau$.
Once the size of $\Rp_i$ is at least
$\frac{3}{2}$ the size of $R_i$,
we reconstruct level $i$.
A formal pseudocode is shown in Algorithm \ref{alg:insert}.

\begin{algorithm}[ht]
  \caption{Insert}
  \label{alg:insert}
  \begin{algorithmic}[1]
    \Procedure{Insert}{$v$}
       \State $\Rp_{0} \gets \Rp_0 \cup \{v\}$
      \For{$i\gets 1, \dots, T + 1$}
        \If{$f(v|G_{i-1}) < \tau$ or $|G_{i-1}| = k$}
          \State \Break\label{line:insert_break}
        \EndIf
        \State $\Rp_i \gets \Rp_i \cup \{v\}$ \label{line:insert_add}
        \If{$i=T + 1$ or $|\Rp_{i}| \ge \frac{3}{2}\cdot |R_i|$}
          \State \ReconstructF{}(i)\label{line:insert_level}
          \State \Break
        \EndIf
      \EndFor
    \EndProcedure
  \end{algorithmic}
\end{algorithm}

\subsection{Deletion}
\label{sec:alg_delete}
To handle deletions,
we keep track of the deleted elements
by adding them to a set $D$.
For each level $i$, we keep track of the number of elements
deleted in the bucket $R_i^{(b(i))}$ that we had chosen
to sample $S_i$ from. 
Once an $\epsilon$-fraction of
these elements is deleted, we restart the offline algorithm from level $i$
by invoking \ReconstructF{}(i).
A formal pseudocode is provided in Algorithm \ref{alg:delete}.
Note that level $i$ is reconstructed with $R_i \gets \Rp_i \backslash D$
to avoid using deleted values in the reconstruction.

\begin{algorithm}[ht]
  \caption{Delete}
  \label{alg:delete}
  \begin{algorithmic}[1]
    \Procedure{Delete}{$v$}
      \State $D \gets D \cup v$
      \For{$i\gets 1, \dots, T$}
        \If{$|D \cap R_i^{(b(i))}| \ge \epsDel \cdot |R_i^{(b(i))}|$}
          \State \ReconstructF{}(i)
          \State \Break
        \EndIf
      \EndFor
    \EndProcedure
  \end{algorithmic}
\end{algorithm}

\subsection{Choice of sample size}
\label{sec:alg_sample}
As explained in the intro,
the main idea behind 
our algorithm is to find the largest integer $m_i$
such that if we sample $m_i$ elements uniformly at random, the last sampled element
 has marginal gain at least
$\tau^{(i)}$ with probability at least $1-\epsilon$.
To achieve this, we binary search over all possible values $m'$.
For each $m'$, we use repeated trials to estimate the probability 
that the last element of $S_i$ has marginal gain at least $\tau^{(i)}$.
As we show in Section \ref{sec:theory_sample_count}, a standard Chernoff bound implies that using polylogarithmic number of samples, we can estimate this probability with error at most $\frac{\epsSamp}{10}$.
A formal pseudocode of the above sketch is provided in Algorithm \ref{alg:sample} 
in which we use notation \underline{u.a.r} for a  set that is sampled \emph{uniformly at random}.

\begin{algorithm}[ht]
  \caption{\CalcSampleCountF{}}
  \label{alg:sample}
  \begin{algorithmic}[1]
    \Function{\ReduceMeanF{}}{$R', G', m'$}
      \For{$t \gets 1, \dots, \left\lceil\frac{4}{\epsSamp^2}\log\left(\frac{200k^{11}}{\epsSamp}\right)\right\rceil $}
        \State Sample $S'$ of size $m' - 1$ u.a.r. from $R'$
        \State Sample element $s$ from $R' \backslash S'$ u.a.r. 
        \State $I_t \gets \ind{f(s | G' \cup S') \ge \tau'}$
      \EndFor
      \State \Return mean of $I_t$
    \EndFunction
    \Function{CalcSampleCount}{$R', G', \tau'$}
      \State $m \gets 1, M \gets \min\{k - |G'|, |R'|\}$
      \If{$\ReduceMeanF{}(R', G', M) \ge 1 - \epsSamp$}
        \State \Return $M$
        \label{line:output_M_ajk}
      \EndIf
      \While{$M - m > 1$}
        \State $m' \gets  {\lfloor\frac{m + M}{2}\rfloor}$
        \If{$\ReduceMeanF{}(R', G', m') \ge 1 - \epsSamp$}
          \State $m \gets m'$
        \Else
          \State $M \gets m'$
        \EndIf
      \EndWhile
      \State \Return $m$
      \label{line:output_m_ajk}
    \EndFunction
  \end{algorithmic}
\end{algorithm}

\subsection{Unknown \opt{}}
\label{sec:alg_parallel}
In this section, we relax the assumption of known \opt{} by maintaining multiple parallel
runs, indexed by $p \in \mathbb{Z}$.
For each run $p$, we will use the value
$\opt_p = (1+\epsOpt)^p$ for $\opt$ in the algorithm
and only insert elements with $f(e) \in [\epsilon \cdot \frac{\opt_p}{2k}, \opt_p]$.
We always output the set with maximum value of $f$ across all the runs.

\subsection{Choice of parameters}
\label{sec:choice_param}
For some $\epsilon< \frac{1}{10}$, we set
$\epsSamp=\epsBuck=\epsOpt = \epsilon$ and
$\epsDel = \frac{\epsilon}{20}$. 

\section{Theoretical analysis}
In this section, we state our main theoretical result.
\begin{theorem}\label{thm:main}
  There is an algorithm for dynamic submodular maximization
  that maintains a set with expected $\frac{1}{2} - \epsilon$ approximation factor
  that makes
  at most $\poly(\log(n), \log(k), \frac{1}{\epsilon})$ amortized queries in expectation,
  where $n$ denotes the largest number of elements at any point of the stream.
\end{theorem}
To prove this result, we establish two theorems that
consider the approximation factor and query complexity of our algorithm respectively.
First, in Section \ref{sec:proof_approx},
we prove the following result.
\begin{theorem}\label{thm:approx}
  The dynamic Algorithm \ref{alg:offline} (with Algorithms \ref{alg:insert} and \ref{alg:delete} for handling
  updates) maintains an output set $G_T \backslash D$ with expected approximation factor
  \begin{math}
    \frac{1}{2} - O(\epsilon)
  \end{math}
  as long as the optimal solution $\optSet$ satisfies
  \begin{align*}
      f(\optSet) \le \opt \le (1+\epsOpt) f(\optSet).
  \end{align*}
\end{theorem}
Next, in Section \ref{sec:query}, we prove the following result.
\begin{theorem}\label{thm:query}
  The dynamic algorithm \ref{alg:offline} (with Algorithms \ref{alg:insert} and \ref{alg:delete} for handling
  updates) has an expected amortized query complexity that is
  polynomial in $\log(n), \log(k),$  and $\frac{1}{\epsilon}$.
\end{theorem}
We note that the above results
imply Theorem \ref{thm:main} since,
as mentioned in Section \ref{sec:alg_parallel},
maintaining multiple runs adds a $\log_{1+\epsOpt}(k/\epsilon)$ term in the query complexity while the approximation ratio reduces by an extra $\epsilon$ term.

\section{Proofs}
\label{appendix:theory}
 We establish some properties of our algorithm in Section \ref{sec:invariant} and \ref{sec:theory_sample_count}.
Throughout the section, we use $\Rc_{i}$ to denote $\Rp_i \backslash D$.
 
 \subsection{Invariants}
 \label{sec:invariant}
 \subsubsection{Deterministic invariants}
In this section, we prove some invariants that will be useful in our proofs.
We start with the following fact.
\begin{fact} \label{fact:level_affect_i_le_j}
  For all $1 \le i < j$, a call to $\ReconstructF{}(j)$ would not affect any of the values
  $R_{i}, \Rp_i, G_i$.
  Furthermore, the value of $\Rc_i := \Rp_i \backslash D$ is unaffected by a call to $\ReconstructF{}(i)$.
\end{fact}
The fact follows from the pseudocodes given for \ReconstructF{}.
\begin{lemma}\label{lm:invariant_filter}
  For all $i \in [0, T]$,
  defining $\Rc_i := \Rp_i \backslash D$,
  \begin{equation}
    \Rc_{i + 1} = \FilterF{}(\Rc_i, G_i).
    \label{eq:invariant_filter}
  \end{equation}
\end{lemma}
\begin{proof}
  Fix $i$ and consider the last time $\ReconstructF{}(j)$ was called
  for some $j \le i$.
  Our proof consists of 2 parts. We
  first show that \eqref{eq:invariant_filter} holds right after this call was made (Part 1).
  We then show that the condition holds after any subsequent insertion and deletion operations (Part
  2).

  \textbf{Part 1.}
  Consider the last execution of $\ReconstructF{}(j)$ for some $j \le i$.
  The following properties hold after this execution.
  \begin{enumerate}
    \item $R_{i+1} = \FilterF{}(R_i, G_i)$ because of Line \ref{line:level_filter}.
    \item $\Rp_i = R_i$ and $\Rp_{i+1} = R_{i+1}$ because of Line \ref{line:level_filter}
    (or line~\ref{line:level_first_line} for $i$, if $j = i$). 
    \item $R_i \cap D = R_{i+1} \cap D = \emptyset$
      because $D$ is removed from $R_j$ in Line \ref{line:level_first_line}, and
      all $R_{j'}$ for $j' > j$ are subsets of $R_j$
      because of Line \ref{line:level_filter}.
  \end{enumerate}
  Therefore,
  \begin{equation*}
    \Rc_{i+1} = R_{i+1} = \FilterF{}(R_i, G_i) = \FilterF{}(\Rc_{i}, G_i)
  \end{equation*}

  \textbf{Part 2 (Update).}
  We assume that the update does not trigger a call to
  \ReconstructF{}(j) for $j \le i$ as otherwise,
  the claim holds by Part 1.
  ~\\
  We first consider \InsertF{}.
  Each time a new value $v$ is inserted to
  $\Rp_i$, then it is inserted into $\Rp_{i+1}$ if and only if
  $f(v | G_{i})$ is at least  $\tau$ and
  $|G_{i-1}| < k$.
  Therefore, $\Rc_{i+1} = \FilterF(\Rc_{i}, G_{i})$
  holds after the execution of Line \ref{line:insert_add} for $i+1$.
  By Fact \ref{fact:level_affect_i_le_j},
  even if the insertion triggers $\ReconstructF{}(i + 1)$,
  the values of
  $\Rc_{i}$ and $\Rc_{i+1}$ would not be affected. Therefore,
  \InsertF{} preserves \eqref{eq:invariant_filter}.
  ~\\
  As for \DeleteF{},
  when an element is added to $D$, it is removed from both
  $\Rc_i$ and $\Rc_{i+1}$,
  and therefore \eqref{eq:invariant_filter} is preserved.
  The possible calls to $\ReconstructF{}(j)$ for
  $j \ge i+1$ also preserve the property by Fact \ref{fact:level_affect_i_le_j}, which finishes the proof.
\end{proof}
\begin{lemma}\label{lm:final_level_empty}
  \begin{math}
    R_{T+1} = \emptyset \label{eq:invariant_final}
  \end{math}
  and $R_{i} \ne \emptyset$ for $i \in [1, T]$. 
\end{lemma}
\begin{proof}
  The claim holds at the beginning of the stream as $\FilterF{}(\emptyset, \cdot) = \emptyset$.
  ~\\
  For each update in the stream,
  if the update triggers a call to $\ReconstructF{}(i)$ for some $i \le T + 1$,
  the property would be preserved since $\ReconstructF{}$ 
  stops building levels when $R_{T+1} = \emptyset$.
  Otherwise, we observe that since
  the value of $T$ and $R_i$ can only change
  during insertion and deletion through invoking $\ReconstructF{}$, the invariant will be preserved.
\end{proof}
\begin{lemma}
\label{lm:reconstruction_condition}
    For all $i \in [T]$
    \begin{align*}
        |R_i^{(b(i))} \cap D|
        \le 
        \epsDel|R_i^{(b(i))}|.
    \end{align*}
    and
    \begin{align*}
        |\Rp_i| \le \frac{3}{2}|R_i|. 
    \end{align*}
\end{lemma}
\begin{proof}
    The claim holds initially because $T=0$, and it holds {after each update} because of the reconstruction condition of the levels. 
\end{proof}
\begin{lemma}
\label{lm:final_level_stronger}
  For $i \in [T]$,
  $\Rp_i\ne \emptyset$ and
  $\Rc_i \ne \emptyset$ and
  $\Rp_{T+1} = \Rc_{T+1} = \emptyset$.
\end{lemma}
\begin{proof}
  We use Lemma \ref{lm:final_level_empty} in both cases.
    We note that for $T+1$, we have
    $|\Rp_{T+1}| \le \frac{3}{2} |R_{T+1}| = 0$ given Lemma \ref{lm:reconstruction_condition}.
    Therefore,
    $\Rp_{T+1} = \Rc_{T+1}=\emptyset$.

    For $i \in [T]$, we observe that
    if $\Rc_i = \emptyset$, then it follows that
    $R_i^{(b(i))} \subseteq D$, which implies
    $|R_i^{(b(i))} \cap D| = |R_i^{(b(i))}|$.
    According to Lemma \ref{lm:reconstruction_condition}, this is only possible when
    $|R_i^{(b(i))}|=0$.
    Because $b(i)$ was the largest bucket, 
    this implies
    that $|R_i|=0$ 
    which contradicts Lemma \ref{lm:final_level_empty}. 
\end{proof}
\begin{lemma}\label{fact:rp_subset}
  For all $i \in [0, T]$,
  $\Rp_{i+1} \subseteq \Rp_i$.
\end{lemma}
\begin{proof}
  The claim holds initially since $T=0$.
  It suffices to show that \InsertF{} and \DeleteF{} preserve it.
  
  We first observe that each time \ReconstructF{}(j)
  is called for some $j \le i$, the property will hold since
  \begin{equation*}
    \Rp_{i+1} = R_{i+1} = \FilterF{}(R_i, G_i) \subseteq R_i = \Rp_i
  \end{equation*}
  It is also clear that calling $\ReconstructF{}(j)$ for $j\ge i+1$ will preserve the property
  since it does not alter $\Rp_i$, while it may (if $j=i+1$)
  decrease $\Rp_{i+1}$. Other than $\ReconstructF{}$,
  the only way either $\Rp_i$ or $\Rp_{i+1}$ are altered is when an element is inserted
  into the sets during processing \InsertF{}, but this also preserves the property since
  an element that was not inserted into $\Rp_i$, will not be inserted into $\Rp_{i+1}$.
\end{proof}
\subsubsection{Random invariants}
Next, we will state and prove the uniformity invariant. 
Before we do this though, we will need a few definitions.
\begin{definition}[History]\label{def:history}
  We define the \emph{History} of Level $i$ as
  \begin{equation}
    H_i := (\Rp_0, \Rp_1, \dots, \Rp_i, R_0, R_1, \dots, R_i, S_1, \dots, S_{i-1}, m_i).
    \label{eq:def_history}
  \end{equation}
  Note that $S_i$ is not included in definition.
  Analogously, we define the random variable $\bH_i$ as 
  \begin{equation*}
    \bH_i := (\bRp_0, \bRp_1, \dots, \bRp_i, \bR_0, \bR_1, \dots, \bR_i, \bS_1, \dots, \bS_{i-1}, \bM_i).
  \end{equation*}
\end{definition}
\begin{remark}
  In the above definition, the random variable $\bH_i$ is defined if and only if $\bT \ge i$.
  Note that $\bH_i$ is not
  defined for $i=T+1$ because
  there is no value of $m_i$.
  This will not be an issue in our analysis since we will work with $\bH_i$ only when $\bT \ge i$. 
  Most notably, when conditioning on the value of $\bH_i$, we will condition on the event $\bT \ge i$ as well.
  Throughout the proofs, we will frequently write $\sigma_i$ instead of
  $\bT \ge i, \bH_i=H_i$ for convenience.
\end{remark}
Our next lemma shows that conditioned on the history of level $i$, its sample set $S_i$
is a random subset of size $m_i$ from $R_i^{(b(i))}$.
\begin{lemma}\label{lm:uniform_lazy}
  For any $i\ge 1$, 
  and any $H_i$ such that
  $\Pr{\bT \ge i, \bH_i = H_i} > 0$,
  \begin{equation}
    \Pr{\bS_i = S | \bT \ge i, \bH_i = H_i} = \frac{1}{\binom{|R_i^{(b(i))}|}{m_i}}\ind{S \subseteq R_i^{(b(i))} \land |S| = m_i} 
    .
    \label{eq:invariant_uniform}
  \end{equation}
\end{lemma}
Before stating the proof, we observe that the property holds after a call to \ReconstructF{}, as we formally show in the following 
Lemma.
\begin{lemma}\label{lm:uniform_right_after}
  Assume we call $\ReconstructF{}(j)$ for some $j\le i$ in a data structure with values
  $T^{-}, R_0^-, \dots$, satisfying $T^- \ge i$,
  obtaining the new (random) values
  $\bT, \bR_0, \dots$.
  Then for all sets $S$,
  \begin{equation*}
    \Pr{\bS_i = S | \bT \ge i, \bH_i = H_i} = \frac{1}{\binom{|R_i^{(b(i))}|}{m_i}}\ind{S \subseteq R_i^{(b(i))} \land |S| = m_i} 
    .
  \end{equation*}
\end{lemma}
\begin{proof}
We observe that
before $S_i$ is sampled
in Line \ref{line:level_choose_sample},
all of the values compromising $H_i$ are already determined, and will not 
change after $S_i$ is sampled.
Therefore, since the claim holds
when $S_i$ is sampled (because of Line~\ref{line:level_choose_sample}), it holds afterwards as well.
\end{proof}

Next, we prove Lemma \ref{lm:uniform_lazy} by showing that \eqref{eq:invariant_uniform} is
preserved after insertions and deletions.
\begin{proof}[Proof of Lemma \ref{lm:uniform_lazy}]
  We will prove the claim by induction on the time step.
  At the beginning of the stream, the claim holds trivially since
  the event $\bT \ge i$ is impossible because of $T = 0 < 1 \le i$.

  Assuming that the claim holds for some time step, we will show that it holds in the next time step as well.
  More formally, for an insertion or deletion of an element $v$ in the stream,
  let $\bT^{-}, \bR_0^{-}, \dots$ denote the values of the data structure before the operation
  and $\bT, \bR_0, \dots$, denote the values afterwards.
  For a fixed $i$, we need to show that \eqref{eq:invariant_uniform} holds.
  The main idea behind the proof is to consider two cases, based on whether or not $\LevelF{}(i)$ was triggered by the insertion or deletion operation.
  In the first case, we will show that Lemma~\ref{lm:uniform_right_after} implies the claim. 
  In the second case, we will 
  use the induction hypothesis,
  i.e, the fact that \eqref{eq:invariant_uniform} held for the values $\bT^{-}, \bR_0^{-}, \dots$,
  to prove the claim.
  As we will observe, a crucial aspect of our proof in the second case is that the decision to invoke \ReconstructF{}(i) in \InsertF{} and \DeleteF{} procedures is solely determined by the History of level $i$ and \emph{not} influenced by the actual samples $S_i$.

  We now proceed with a formal proof.
  Let $\bL_i$ be a random variable that takes the value $1$ if $\LevelF{}(i)$ was called because of the update,
  and takes the value $0$ otherwise.
  Let the shorthand $\sigma_i$ denote the event
  $\bT \ge i \land \bH_i = H_i$.
  We need to show that
  $\Pr{\bS_i = S | \sigma_i} = \frac{1}{\binom{|R_i^{(b(i))}|}{m_i}} \cdot \ind{S \subseteq R_i^{(b(i))} \land |S| = m_i}$. 
  Conditioning on $\bL_i$,
  we can rewrite $\Pr{\bS_i=S | \sigma_i}$ as
  \begin{equation}
    \Pr{\bS_i = S| \sigma_i} = 
    \Exu{L_i \sim \bL_i | \sigma_i}{
      \Pr{\bS_i = S| \sigma_i, \bL_i=L_i}
    }.
    \label{eq:s_sigma_decompose_l}
  \end{equation}
  It therefore suffices to prove that 
  \begin{equation}
    \Pr{\bS_i = S| \sigma_i, \bL_i=L_i}
    =
    \frac{1}{\binom{|R_i^{(b(i))}|}{m_i}} \cdot \ind{S \subseteq R_i^{(b(i))} \land |S| = m_i},
    \label{eq:jul11_1815}
  \end{equation}
  for both $L_i=0$ and $L_i=1$.
  For $L_i=1$, the claim holds
  by Fact~\ref{lm:uniform_right_after}
  since by definition, $L_i=1$ means \ReconstructF{}(j) was called for some $j \le i$.

  We therefore focus on $L_i=0$.
  We first observe that
  $\sigma_i, \bL_i=0$ implies
  $\bT^- \ge i$. 
  This is because
  $\bL_i=0$ means
  $\ReconstructF{}(j)$ for any $j \le i$ was not called.
  Therefore if $\bT^-$ were strictly less than $i$, then  $\bT$ would
  equal $\bT^-$ (because
  $\ReconstructF{}$ can only be called for values upto $T+1$ and if $\ReconstructF{}$ is not called, then $T$ does not change.)
  which is not possible since $\bT \ge i$.
  ~\\
  Since $\bT^- \ge i$, we can condition on the value of
  the \emph{previous history} of level $i$.
  More formally, define the random variable $\bH_i^-$ as
  \begin{equation*}
    \bH_i^- := (\bRp_0^-, \bRp_1^-, \dots, \bRp_i^-, \bR_0^-, \bR_1^-, \dots, \bR_i^-, \bS_1^-, \dots, \bS_{i-1}^-, \bM_i^-).
  \end{equation*}
  By law of iterated expectation,
  \begin{align}
    \notag
    \Pr{\bS_i = S| \sigma_i, \bL_i=0}
    &=
    \Pr{\bS_i = S| \sigma_i, \bL_i=0, \bT^-\le i}
    \\&=
    \Exu{H_i^- \sim \bH_i^- | \sigma_i, \bL_i=0
    , \bT^-\le i}{
      \Pr{\bS_i = S| \sigma_i, \bL_i=0, \bT^-\le i, \bH_i^- = H_i^-}
    }
    \label{eq:Jul10_2249}
  \end{align}
  where the expectation is taken over all $H_i^-$ with positive probability.

  We now observe that conditioned on $\bT^-\le i, \bH_i^-=H_i^-$, the value of $\bL_i$ always equals $0$.
  This is because $L_i$ is a function of
  $(\Rp_{0}, \dots \Rp_{i}, R_0,\dots, R_i, D)$,
  which is determined by $H_i$. 
  Note that $D$ is deterministic since it contains all the deleted elements in the
  stream and is independent of our algorithm.
  Therefore, since we only consider
  $H_i^-$ with positive probability, we can drop the conditioning on $\bL_i=0$
  in the $\Pr{\bS_i = S| \sigma_i, \bL_i=0, \bT^-\le i, \bH_i^- = H_i^-}$ term of \eqref{eq:Jul10_2249}
  since it is redundant.
  ~\\
  We can similarly drop $\sigma_i$. 
  This is because as $\bL_i=0$, the value of $\bH_i$ is deterministic
  conditioned on $\bH_i^-=H_i^-$. Notably, the values of $\bR_1, \dots, \bR_{i}$
  are going to be $R_1^-, \dots, R_i^-$. The same
  can be said for $\bS_1, \dots, \bS_{i-1}$ and $\bM_i$. As for
  $\bRp_1, \dots, \bRp_{i}$,
  even though their value maybe different from 
  $\Rp_1^-, \dots, \Rp_{i}^-$, it is \emph{still deterministic}
  conditioned on $\bH_i^-=H_i^-$
  as the decision to add elements in Line \ref{line:insert_add} is
  based on the values in $H_i^-$ only.
  Given the above observation, we can rewrite 
  $\Pr{\bS_i = S| \sigma_i, \bL_i=0}$ as
  \begin{align*}
    \Pr{\bS_i = S| \sigma_i, \bL_i=0}
    &=
    \Exu{H_i^- \sim \bH_i^- | \sigma_i, \bL_i=0}{
      \Pr{\bS_i = S| \bT^-\le i, \bH_i^- = H_i^-}
    }
  \end{align*}
  We can further replace $\bS_i=S$ with $\bS_i^-=S$ as
  $\bL_i=0$ implies $\bS_i = \bS_i^-$. 
  Denoting the largest bucket in $R_i^{-}$ with
  $R_i^{-, (b(i))}$,
  it follows that
  \begin{align*}
    \Pr{\bS_i = S| \sigma_i, \bL_i=0}
    &=
    \Exu{H_i^- \sim \bH_i^- | \sigma_i, \bL_i=0}{
      \Pr{\bS_i^- = S| \bT^-\le i, \bH_i^- = H_i^-}
    }
    \\&\overset{(a)}{=}
    \Exu{H_i^- \sim \bH_i^- | \sigma_i, \bL_i=0}{
      \frac{1}{\binom{|R_i^{(-, b(i))}|}{m_i^-}} \cdot \ind{S \subseteq R_i^{(-, b(i))} \land |S| = m_i^-} 
    }
    \\&\overset{(b)}{=}
    \Exu{H_i^- \sim \bH_i^- | \sigma_i, \bL_i=0}{
      \frac{1}{\binom{|R_i^{(b(i))}|}{m_i}} \cdot \ind{S \subseteq R_i^{(b(i))} \land |S| = m_i} 
    }
    \\&=
    \frac{1}{\binom{|R_i^{(b(i))}|}{m_i}} \cdot \ind{S \subseteq R_i^{(b(i))} \land |S| = m_i} 
  \end{align*}
  where for $(a)$, we have used the induction assumption,
  and for $(b)$ we have used the fact that
  $R_i = R_i^{-}$ and $m_i=m_i^{-}$ because of $\bL_i=0$. 
  We have therefore proved \eqref{eq:jul11_1815} 
  for both $L_i=0$ and $L_i=1$, which completes the proof given \eqref{eq:s_sigma_decompose_l}.
\end{proof}
\subsection{Properties of \CalcSampleCountF{}}
\label{sec:theory_sample_count}
In this section, we state the key properties of \CalcSampleCountF{} that will be useful in our proofs.
We begin by defining the notion of Suitable sample count (Definition \ref{def:suitable_sample_count}), which captures
the properties we require in our proofs. We then show that the output of \CalcSampleCountF{} is suitable with high probability (Lemma \ref{lm:bound_countF}).

\begin{definition}[Suitable sample count]
  \label{def:suitable_sample_count}
  Given $R', G', \tau'$ such that
  $R' \ne \emptyset$ and $\FilterF{}(R', G', \tau') = R'$,
  define their \emph{Suitable sample count} $M^*(R', G', \tau')$ as the set of all integers like $m'$ such that
  $m' \ge 1$, and
  if we sample a set $S'$ with $m' - 1$ elements from $R'$ and an element $s'$ from 
  $R' \backslash S'$,
  \begin{equation}
    \Ex{\abs{\FilterF{}(R', G' \cup S' \cup s', \tau')}} \le (1-\epsSamp/4)|R'|
    \label{eq:sample_count_filter}
  \end{equation}
  and
  \begin{equation}
    \Pr{f(s' | G' \cup S') \ge \tau'} \ge (1-2\epsSamp)
    \label{eq:sample_count_quality}
  \end{equation}
  Define the suitable sample count of level $i$ as
  $M_i^* := M^*(R_i^{(b(i))}, G_{i-1}, \tau^{(i)})$.
\end{definition}

\begin{lemma}\label{lm:bound_countF}
  Consider a call to $\CalcSampleCountF{}(R', G', \tau')$
  with values satisfying $\FilterF{}(R', G', \tau') = R'$ 
  and $R' \ne \emptyset$.
  The number of queries made by \CalcSampleCountF{} is bounded by
  \begin{math}
    \mO\left(|R'| \cdot \frac{\log(k)}{\epsSamp^3}\right).
  \end{math}
  Furthermore, the output is a suitable sample count (Definition \ref{def:suitable_sample_count}) with probability at least $1-\frac{\epsSamp}{100k^{10}}$.
\end{lemma}
\begin{proof}
  To bound the number of queries, we note that
  there will be at most $\min\{|R'|, \log(k)\} \le |R'|$ 
  calls to \ReduceMeanF{}, and each call makes at most
  \begin{align*}
    \mO\left(\frac{4}{\epsSamp^2}\log(200k^{10}/\epsSamp)\right)
    = 
    \mO\left(\frac{1}{\epsSamp^2}\cdot \left(\log(k) + \log(\frac{1}{\epsSamp})\right)\right),
  \end{align*}
  queries, which implies the lemma's statement.

  We therefore focus on proving that the output is suitable with probability at least
  $1 - \frac{\epsSamp}{100k^{10}}$.

  For a number $m'$, 
  define the value $I_{m'}$ as
  \begin{equation*}
    \Pr{f(s' | G' \cup S') \ge \tau'}.
  \end{equation*}
  where $s', S'$ are sampled as in the Definition \ref{def:suitable_sample_count}, i.e.,
  $S'$ has size $m'-1$ and is sampled from $R'$, and $s'$ is an element sampled from
  $R' \backslash S'$. Observe that $I_{m'}$ is the expectation of the Bernoulli random variable
  $\ind{f(s' | G' \cup S') \ge \tau'}$, while
  $\ReduceMeanF{}(R', G', m')$ estimates this mean by repeated sampling.
  By Hoeffding's inequality, we conclude that for each call to $\ReduceMeanF{}(R', G', m')$,
  \begin{align*}
    \Pr{\abs{I_{m'} - \ReduceMeanF{}(R', G', m')} \ge \epsSamp/2} &\le
    2e^{-\left\lceil \frac{4}{\epsSamp^2}\log(200k^{11}/\epsSamp) \right\rceil \cdot (\frac{\epsSamp}{2})^2}
    \\&\le
    2e^{-\frac{4}{\epsSamp^2}\log(200k^{11}/\epsSamp)\cdot (\frac{\epsSamp}{2})^2}
    \\&\le
    2e^{-\log(200k^{11}/\epsSamp)}
    \\&\le
    \frac{\epsSamp}{100k^{11}}
  \end{align*}
  The above formula is an upper bound on the probability that a single
  call to \ReduceMeanF{} is off by at most $\epsSamp/2$. 
  Since \CalcSampleCountF{} makes at most $M\le k$ calls, a union bound implies that
  with probability at least $1-\frac{\epsSamp}{100k^{10}}$,
  \ReduceMeanF{} is off by at most $\frac{\epsSamp}{2}$ for \emph{all} the calls
  made in \CalcSampleCountF{}. In other words,
  \begin{align}
      \Pr{\forall m' \in [1, \min\{k - |G'|, |R'|\}: \abs{I_{m'} - \ReduceMeanF{}(R', G', m') }\ge \frac{\epsSamp}{2}} \le \frac{\epsSamp}{100k^{10}}.
      \label{eq:17ajak1}
  \end{align}
  For the rest of the proof, we will
  show that whenever
  all of the
  calls to \ReduceMeanF{}
  are off by at most $\frac{\epsSamp}{2}$,
  the output is guaranteed to be a suitable sample count.
  This implies the Lemma's statement because of \eqref{eq:17ajak1}. 

  If $M$ is outputted in Line~\ref{line:output_M_ajk} of Algorithm~\ref{alg:sample}, we conclude that $\Pr{I_M} \ge 1-\frac{3}{2}\epsSamp$,
  which means $M$ satisfies \eqref{eq:sample_count_quality}.
  Note however that $M$ always satisfies \eqref{eq:sample_count_filter}
  since
  $\FilterF{}(R', G' \cup S' \cup s', \tau')$ will always equal $\emptyset$.
  ~\\
  We therefore assume that the output was some value $m' < M$ in Line~\ref{line:output_m_ajk} of Algorithm~\ref{alg:sample}. 
  We note that $\ReduceMeanF{}(m') \ge 1 - \epsSamp$. 
  This holds because all of the values that $m$ can take in Algorithm \ref{alg:sample} satisfy this and the output is $m$; the value $m=1$ satisfies this cause
  $I_{1}=1$ and the others satisfy this given the condition for setting $m$.
  Since $\ReduceMeanF{}(m')$ was accurate within $\epsSamp/2$
  and $\ReduceMeanF{}(m') \ge 1 - \epsSamp$,
  it follows that
  $I_{m'} \ge 1-\frac{3}{2}\epsSamp$, which implies \eqref{eq:sample_count_quality}.

  It remains to show \eqref{eq:sample_count_filter}.
  Sampling $S'$ and $s'$ as before, 
  \begin{align*}
    \Exu{S', s'}{
      \abs{
        \FilterF{}(R', G' \cup S' \cup s', \tau')
        }
      }
    &=
    \Exu{S', s'}{
      \sum_{e \in R'}
        \ind{e \in \FilterF{}(R', G' \cup S' \cup s', \tau')}
      }
    \\&\overset{(a)}{=}
    \Exu{S', s'}{
      \sum_{e \in R' \backslash (S' \cup s')}
        \ind{e \in \FilterF{}(R', G' \cup S' \cup s', \tau')}
      }
  \end{align*}
  where for $(a)$, we have used the fact that $\FilterF{}(A, B, \tau) \cap B$
  is $\emptyset$ for any sets $A$ and $B$.
  Letting $e'$ be a random sample from $R' \backslash (S' \cup s')$,
  this implies that
  \begin{align*}
    \Exu{S', s'}{
      \abs{
        \FilterF{}(R', G' \cup S' \cup s, \tau')
        }
      }
    &=
    \mathbb{E}_{S', s'}\Bigg[{
      |R' \backslash (S' \cup s')| \cdot
      \mathbb{E}_{e'}\Big[{
        \ind{e'\in \FilterF{}(R', G' \cup S' \cup s', \tau')}
        }\Big]
      }\Bigg] 
    \\&\le
    |R'| \cdot
    \Exu{S', s'}{
      \Exu{e'}{
        \ind{e'\in \FilterF{}(R', G' \cup S' \cup s', \tau')}
        }
      }
    \\&=
    |R'| \cdot
    \Pru{S', s', e'}{
        e'\in \FilterF{}(R', G' \cup S' \cup s', \tau')
      }
    \\&\overset{(a)}{=}
    |R'| \cdot
    \Pru{S', s', e'}{
        f(e'| G' \cup S' \cup s') \ge \tau'
      }
  \end{align*}
  {where for $(a)$, we have used the fact that
  since $m' < M \le k - |G'|$, we have $|G' \cup S' \cup s'| < k$, which implies $e'\in \FilterF{}(R', G' \cup S' \cup s', \tau')$ is equivalent to 
  $f(e'| G' \cup S' \cup s') \ge \tau'$.
  }
  Note however that since $S' \cup s'$ is a random subset of size $m'$
  and $e'$ is a random sample of $R' \backslash (S' \cup s')$, we have
  $ \Pru{S', s', e'}{ f(e | G' \cup S' \cup s') \ge \tau' }$ equals
  $I_{m'+1}$ (note that $m' + 1 \le M$ because we had assumed earlier that $m'< M$). Since $m'\ne M$, we conclude that
  $\ReduceMeanF{}(R', G', m'+1) \le 1-\epsSamp$, which implies
  $I_{m'+1} \le 1-\frac{\epsSamp}{2}$, which implies \eqref{eq:sample_count_filter}.
\end{proof}
The above lemma effectively shows that $m_i \in M_i^*$
holds with high probability \emph{right after $\ReconstructF{}(j)$}
is called for some $j \le i$. 
To prove our approximation guarantee, we will need a stronger result that shows this holds at an arbitrary point in the stream. 
Before stating the result, we define the pre-count history of a level as, effectively,
its history without $m_i$.
\begin{definition}[Pre-count History]\label{def:pre_history}
  We define the \emph{Pre-count History} of Level $i$ as
  \begin{equation}
    H_i^{-m} := (\Rp_0, \Rp_1, \dots, \Rp_i, R_0, R_1, \dots, R_i, S_1, \dots, S_{i-1}).
    \label{eq:def_pre_count_history}
  \end{equation}
  Note that this is similar to the definition of $H_i$ in \eqref{eq:def_history},
  with the difference that $m_i$ is no longer included.
  Analogously, we define the random variable $\bH_i^{-m}$ as 
  \begin{equation*}
    \bH_i^{-m} := (\bRp_0, \bRp_1, \dots, \bRp_i, \bR_0, \bR_1, \dots, \bR_i, \bS_1, \dots, \bS_{i-1}).
  \end{equation*}
\end{definition}
\begin{lemma}\label{lm:bound_countF_always}
  At any point in the stream,
  for any $i\ge 1$,
  \begin{equation*}
    \Pr{\bM_i \notin M_i^* | \bT \ge i, \bH_i^{-m} = H_i^{-m}} \le \frac{\epsSamp}{100k^{10}}
  \end{equation*}
\end{lemma}
\begin{proof}
  The proof follows in the same manner as Lemma \ref{lm:uniform_lazy}.

  We prove the claim by induction on the update stream.
  Initially, the claim holds trivally since $\bT\ge i$ is impossible.
  Assuming the lemma's statement holds before an update, we will show that it holds for the new values as well.
  Let the superscript $^{-}$ to denote the values before the insertion, e.g., $T^{-}$ denotes the number of levels before insertion and $T$ denotes the number of levels after insertion.

  As before, let $\bL_i$ denote a random variable that is set to 1 if
  $\ReconstructF{}(j)$ is called for some $j\le i$ and set to 0 otherwise.
  We will show that
  \begin{equation*}
    \Pr{\bM_i \notin M_i^* | \bT \ge i, \bH_i^{-m} = H_i^{-m}, \bL_i=L_i} \le \frac{\epsSamp}{100k^{10}}
  \end{equation*}
  for both $L_i=0$ and $L_i=1$. For $L_i=1$, the claim follows from
  Lemma \ref{lm:bound_countF}.
  As for $L_i=0$, we note that this implies
  $\bT^- \ge i$ because
  if $\bT^-$ were $< i$,
  then $\bT$ would have been equal to  $\bT^-$ because
  $\bL_i=0$ means that $\ReconstructF{}(j)$ was not invoked for
  any $j\le i$. 
  Let $(H_i^{-m})^{-}$ denote the value of pre-count history before the update. 
  By law of iterated expectation, it suffices to show
  \begin{equation}
    \Pr{\bM_i \notin M_i^* | \bT \ge i, \bT^- \ge i, \bH_i^{-m} = H_i^{-m}, (\bH_i^{-m})^{-} = (H_i^{-m})^{-}, \bL_i=0} \le \frac{\epsSamp}{100k^{10}}
    \label{eq:jul12_1507}
  \end{equation}
  for all $(H_i^{-m})^{-}$ that have positive probability conditioned on $\bT\ge i,\bH_i^{-m}=H_i^{-m}, \bL_i=0$.
  As before, we can drop the $\bT \ge i, \bH_i^{-m}=H_i^{-m}, \bL_i=0$ from the condition
  since they are implied by $(\bH_i^{-m})^{-}=(H_i^{-m})^{-}$. 
  We can further replace $\bM_i \notin M_i^*$ with
  $\bM_i^{-} \notin (M_i^{*})^{-}$, where $(M_i^{*})^{-}$ is the set of suitable sample counts for level $i$ before the update, as determined by $(H_i^{-m})^{-}$.
  This is because both $\bM_i$ and $M_i^*$ are unchanged through the udpate (note that $M_i^*$ is a function of
  $R_i, G_{i-1}$ and $(R_i, G_{i-1})=(R_i^{-}, G_{i-1}^{-}$)).
  This changes \eqref{eq:jul12_1507} to the induction hypothesis,
  which finishes the proof.
\end{proof}
\subsection{Approximation guarantee}\label{sec:proof_approx}
In this section, we prove the approximation guarantee of our algorithm.

Our analysis is based on a novel relaxation function that may be of independent interest.
We first introduce the relaxation function $f'$ and show that 
it
is a lower bound on $f(G_T \backslash D)$ (Lemma \ref{lm:fp_lower_bound}).
We then bound the contribution of each level to the relaxation function (Lemmas \ref{lm:bound_f_one_level} and \ref{lm:bound_f_total}),
to prove the desired $\frac{\opt}{2}-\epsilon$ bound.

We begin with some definitions.
For $1 \le i \le T$, define the quantities $f_i$ and $d_i$ as
\begin{equation*}
  f_i := f(S_i | G_{i-1}),
  \quad \tau^{(i)} := (1+\epsBuck)^{b(i)} \cdot \tau,
  \quad d_i := |D \cap S_i| \cdot (1+\epsBuck)\cdot \tau^{(i)}.
\end{equation*}
We note that
\begin{equation}
  f(G_T) = 
  \sum_{i=1}^{T}
  \paranth{
    f(G_i) - f(G_{i - 1}) 
  }
  =
  \sum_{i=1}^{T} 
    f(S_i | G_{i-1})
  = \sum_{i=1}^{T} f_i.
  \label{eq:f_G_T_equal_sum_f_i}
\end{equation}
We can also bound $f(G_T\backslash D)$ in terms of $f(G_T)$ and $d_i$
as
\begin{align}
  f(G_T \backslash D) 
  =
  \sum_{i=1}^{T}
  \left(
  f(G_i \backslash D) - f(G_{i-1} \backslash D)
  \right)
  &=
  \sum_{i=1}^{T}
  f(S_i \backslash D | G_{i-1} \backslash D)
  \notag
  \\&\ge
  \sum_{i=1}^{T}
  f(S_i \backslash D | G_{i-1})
  \notag
  \\&\overset{(a)}{\ge}
  \sum_{i=1}^{T}
  \left(
  f(S_i| G_{i-1})
  - f(D \cap S_i| G_{i-1})
  \right)
  \notag
  \\&\ge 
  \sum_{i=1}^{T}
  f(S_i | G_{i-1})
  - \sum_{i=1}^T \sum_{e \in D \cap S_i} f(e | G_{i-1})
  \notag
  \\&\overset{(b)}{\ge}
  \sum_{i=1}^T f_i - \sum_{i=1}^{T} d_i
  \label{eq:sep13_1317}
\end{align}
where the first three inequalities follow from submodularity (in particular, $(a)$ follows from
the fact that $f(A \backslash B) \ge f(A) - f(B)$ for any sets $A, B$)
and $(b)$ follows from the fact that $f(e|G_{i-1})\le (1+\epsBuck)\cdot \tau^{(i)}$
for all $e\in S_i$.

For $i\in [1, T+1]$, Define $f'_i$ as 
\begin{align*}
  f'_i =
  \begin{cases}
    \min \{f_i, m_i \cdot \tau^{(i)} \}  - d_i
    \CasesIf
    i \le T \\
    (1-4\epsSamp) \left(
      \frac{1-\epsOpt}{1+\epsOpt}\frac{\opt}{2} - \sum_{j\le T} m_j \cdot \tau^{(j)}
      \right)
    \CasesIf i = T + 1
  \end{cases}
\end{align*}
Our next lemma shows that $\sum_{i=1}^{T + 1}f'_i$ is a lower bound for the value of output.
\begin{lemma}\label{lm:fp_lower_bound}
  Assume that
  \begin{equation*}
    f(\optSet) \le \opt \le (1+\epsOpt) f(\optSet).
  \end{equation*}
  Then
  \begin{align*}
    f(G_T \backslash D) \ge \sum_{i=1}^{T+ 1}f'_i.
  \end{align*}
\end{lemma}
\begin{proof}
  We divide the proof into two cases, depending on whether or not $f'_{T+1} > 0$.
  \paragraph{Case 1:}
  $f'_{T+1} \le 0$.
  In this case,
  by \eqref{eq:sep13_1317},
  \begin{align*}
    f(G_T \backslash D)
    \ge \sum_{i=1}^{T}f_i - \sum_{i=1}^{T} d_i
    \ge \sum_{i=1}^{T} \left( \min\{f_i, m_i \cdot \tau^{(i)}\} - d_i\right)
  = \sum_{i=1}^{T}f'_i
  \ge \sum_{i=1}^{T+1}f'_i \enspace .
  \end{align*}
  Where the last inequality uses the assumption $f'_{T+1} \le 0$.
  \paragraph{Case 2:}
  $f'_{T+1} > 0$.
  We first claim that
  $|G_T| < k$:
  \begin{align}
    f'_{T+1} > 0 \implies 
    \frac{1-\epsOpt}{1+\epsOpt}\frac{\opt}{2} - \sum_{j=1}^T m_j \cdot \tau^{(j)} > 0
    .
    \label{eq:jul11_1105}
  \end{align}
  which further implies {that}
  \begin{align*}
    k\tau \overset{(a)}{=}
    \frac{\opt}{2} 
    \ge
    \frac{1-\epsOpt}{1+\epsOpt}\frac{\opt}{2}
    \overset{(b)}{>}
    \sum_{j=1}^{T} m_j \cdot \tau^{(j)} 
    \ge \sum_{j=1}^{T} m_j \cdot \tau 
    = |G_T| \cdot \tau.
  \end{align*}
  Here, we have used the definition of $\tau$ for $(a)$ and
  \eqref{eq:jul11_1105} for $(b)$.
  Therefore, $|G_T| < k$ as claimed.

  Consider an element $s\in \optSet$.
  By Lemma \ref{lm:final_level_stronger},
  $\Rc_{T+1}=\emptyset$ and therefore
  $s\notin \Rc_{T+1}$. 
  Since $s\in \Rc_0$, there is an
  index $i \in [0, T]$ such that
  $s \in \Rc_i$ but $s \notin \Rc_{i+1}$.
  By Lemma \ref{lm:invariant_filter}, this means that either
  $f(s | G_{i}) < \tau$ or $|G_i| = k$. Since $|G_i| \le |G_T| < k$,
  we conclude that $f(s| G_{i}) < \tau$, which by submodularity implies
  $f(s | G_{T}) < \tau$. 
  Now note that
  \begin{align*}
    \frac{\opt}{1+\epsOpt}
    \le 
    f(\optSet)
    \le
    f(G_T \cup \optSet) 
    \le
    f(G_T) + \sum_{s^* \in \optSet}f(s^* | G_T)
    \le
    f(G_T) + k \cdot \tau 
    = f(G_T) + \frac{\opt}{2} \enspace .
  \end{align*}
  Rearranging and using \eqref{eq:f_G_T_equal_sum_f_i}, we obtain
  \begin{equation*}
    \sum_{i=1}^{T} f_i =
    f(G_T)
    \ge
    \frac{\opt}{1+\epsOpt} - 
    \frac{\opt}{2}
    = 
    \frac{\opt(2 - (1+\epsOpt))}{2(1+\epsOpt)}
    =
    \frac{\opt}{2}\cdot (\frac{1-\epsOpt}{1+\epsOpt}).
  \end{equation*}
  Since in \eqref{eq:sep13_1317} we show that $f(G_T \backslash D) \ge \sum_{i=1}^{T} f_i - \sum_{i=1}^{T} d_i$, this implies that
  \begin{align*}
    f(G_T \backslash D)
    &\ge 
    \frac{1-\epsOpt}{1+\epsOpt}\frac{\opt}{2}
    - \sum_{i=1}^{T} d_i
    \\&\overset{(a)}{=}
    \sum_{i=1}^{T} \left(m_i \cdot \tau^{(i)} - d_i \right)
    +\left(\frac{1-\epsOpt}{1+\epsOpt}\frac{\opt}{2}
    - \sum_{i=1}^{T} m_i \cdot \tau^{(i)}\right)
    \\&\ge
    \sum_{i=1}^{T}\left( \min\{f_i, m_i \cdot \tau^{(i)}\} - d_i 
    \right)
    +\paranth{\frac{1-\epsOpt}{1+\epsOpt}\frac{\opt}{2}
    - \sum_{i=1}^{T} m_i \cdot \tau^{(i)}}
    \\&\overset{(b)}{\ge}
    \sum_{i=1}^{T} \paranth{\min\{f_i, m_i \cdot \tau^{(i)}\} - d_i }
    +(1-4\epsSamp)
    \left(
    \frac{1-\epsOpt}{1+\epsOpt}\frac{\opt}{2}
    - \sum_{i=1}^{T} m_i \cdot \tau^{(i)}
    \right)
    \\&= \sum_{i=1}^{T+1} f'_i
  \end{align*}
  where for $(a)$ we have just added and subtracted $\sum_{i=1}^{T} m_i \cdot \tau^{(i)}$, and for $(b)$, we have used the 
  assumption $f'_{T+1} > 0$.
\end{proof}
\begin{lemma}\label{lm:bound_f_one_level}
  For all $i$, if $m_i \in M_i^*$, then
  \begin{align*}
    \Ex{f'_i\big| \bT \ge i, \bH_i = H_i} \ge (1-3\epsSamp) \cdot m_i \cdot \tau^{(i)}
  \end{align*}
  for all $H_i$ such that $\Pr{\bT \ge i, \bH_i = H_i} > 0$. 
\end{lemma}
\begin{proof}
  Since $\Pr{\bT \ge i, \bH_i = H_i} > 0$, we conclude that $R_i \ne \emptyset$. 
  By Lemma \ref{lm:uniform_lazy},
  $S_i$ is a uniform sample of size $m_i$ from $R_i^{b(i)}$.
  Let $s_{i, 1}, \dots, s_{i, m_i}$
  be a random ordering of $S_i$
  and let $S_{i, < j}$ and $S_{i, > j}$ denote
  $\{s_{i, 1}, \dots, s_{i, j-1}\}$ and
  $\{s_{i, j + 1}, \dots, s_{i, m_i}\}$ respectively.
  It follows that
  \begin{align*}
    f_i
    &=
    \sum_{j=1}^{m_i}
    f(s_{i, j} | G_{i-1} \cup S_{i, <j})
    \\&\ge
    \sum_{j=1}^{m_i}
    f(s_{i, j} | G_{i-1} \cup S_{i, <j} \cup S_{i, > j})
    \\&\ge
    \sum_{j=1}^{m_i}
    \min\{\tau^{(i)}, f(s_{i, j} | G_{i-1} \cup S_{i, <j} \cup S_{i, > j})\}
  \end{align*}
  It is also clear that
  \begin{equation*}
    m_i \cdot \tau^{(i)}
    \ge
    \sum_{j=1}^{m_i}\min\{\tau^{(i)}, f(s_{i, j} | G_{i-1} \cup S_{i, <j} \cup S_{i, > j})\}
  \end{equation*}
  As before, using the shorthand $\sigma_i$ to denote
  $\bT \ge i, \bH_i = H_i$,
  \begin{align*}
    \Ex{
      \min\{f_i, m_i\tau^{(i)}\}
      \Big| \sigma_i
      }
    &\ge 
    \sum_{j=1}^{m_i}
    \Ex{
    \min\{\tau^{(i)}, f(s_{i, j} | G_{i-1} \cup S_{i, <j} \cup S_{i, > j})\}
    \Big| \sigma_i
    }
    \\&\ge 
    \sum_{j=1}^{m_i}
    \Ex{
    \tau^{(i)}\ind{f(s_{i, j} | G_{i-1} \cup S_{i, <j} \cup S_{i, > j}) \ge \tau^{(i)}}
    \Big| \sigma_i
    }
    \\&=
    \sum_{j=1}^{m_i}
    \tau^{(i)}
    \Ex{
    \ind{f(s_{i, j} | G_{i-1} \cup S_{i, <j} \cup S_{i, > j}) \ge \tau^{(i)}}
    \Big| \sigma_i
    }
    \\&=
    \sum_{j=1}^{m_i}
    \tau^{(i)}
    \Pr{
    f(s_{i, j} | G_{i-1} \cup S_{i, <j} \cup S_{i, > j}) \ge \tau^{(i)}
    \Big| \sigma_i
    }
    \\&\overset{(a)}{=}
    m_i \cdot \tau^{(i)}\cdot 
    \Pr{ f(s_{i, m_i} | G_{i-1} \cup S_{i, <m_i})\ge \tau^{(i)} \Big| \sigma_i}
  \end{align*}
  where $(a)$ follows from the fact that $s_{i, 1}, \dots, s_{i, m_i}$
  was a random permutation.
  ~\\
  We now note that by Lemma \ref{lm:uniform_lazy},
  $S_{i, <m_i}$ is a random subset of size
  $m_i-1$ from $R_i^{(b(i))}$ and
  $s_{i, m_i}$ is a random element in $R_i^{(b(i))} \backslash S_{i, < m_i}$.
  Therefore, given the assumption $m_i \in M_i^*$,
  \eqref{eq:sample_count_quality} implies that
  \begin{align*}
    \Pr{ f(s_{i, m_i} | G_{i-1} \cup S_{i, <m_i})\ge \tau^{(i)} \Big| \sigma_i}
    \ge 1-2\epsSamp.
  \end{align*}
  Therefore,
  \begin{align}
    \Ex{
      \min\{f_i, m_i\tau^{(i)}\}
      \Big| \sigma_i
      }
      \ge (1-2\epsSamp) \cdot m_i \cdot \tau^{(i)}
      .
      \label{lkjqwey}
  \end{align}
  Furthermore, given the definition of $d_i$, we have
  \begin{align}
    \notag
    \Ex{d_i | \sigma_i}
    &=
    \notag
    (1+\epsBuck) \cdot \tau^{(i)}
    \cdot
    \Ex{\abs{S_i \cap D} \Big| \sigma_i}
    \\&\overset{(a)}{=}
    \notag
    (1+\epsBuck) \cdot \tau^{(i)}
    \cdot
    \frac{|R_i^{(b(i))} \cap D|}{|R_i^{(b(i))}|} m_i
    \\& \overset{(b)}{\le}
    (1+\epsBuck) \cdot \tau^{(i)}
    \cdot
    m_i \cdot \epsDel
    .
    \label{yqwerqwu}
  \end{align}
  where $(a)$
  follows from
  the uniform sample assumption of Lemma \ref{lm:uniform_lazy},
  and $(b)$ follows from the restriction $|R_i^{(b(i))} \cap D| \le \epsDel \cdot |R_i^{(b(i))}|$ given in Lemma \ref{lm:reconstruction_condition}.
  We note however that $(1+\epsBuck) \cdot \epsDel \le \epsSamp$ because of the way the parameters were set in Section~\ref{sec:choice_param}.
  Therefore,
  the claim follows from
  \eqref{lkjqwey} and \eqref{yqwerqwu}
  by the linearity of expectation:
  \begin{align*}
    \Ex{
      \min\{f_i, m_i\tau^{(i)}\}
      -d_i
      \Big| \sigma_i
      }
      \ge (1-2\epsSamp) \cdot m_i \cdot \tau^{(i)}
      - \epsSamp\cdot \tau^{(i)} \cdot m_i
      = 
      (1-3\epsSamp) \cdot m_i \cdot \tau^{(i)}
      .
  \end{align*}
\end{proof}
\begin{lemma}\label{lm:bound_f_total}
  For any value $i\ge 1$\footnote{Note that we do not impose the restriction $i \le T + 1$ when specifying the range for $i$ since, effectively, this is done by conditioning  on the event $\bT+1 \ge i$.}, 
  \begin{align*}
    \Ex{\sum_{j= i}^{\bT + 1} f'_j \Bigg| \bT + 1 \ge i, \bH_i^{-m} = H_i^{-m}}
    &\ge
    (1-4\epsSamp) \left(
      \frac{1-\epsOpt}{1+\epsOpt}\frac{\opt}{2} - \sum_{j < i} m_j \cdot \tau^{(j)}
      \right)
  \end{align*}
  for all $H_i^{-m}$ such that 
  $\Pr{\bT + 1 \ge i, \bH_i^{-m} = H_i^{-m}} > 0$.
\end{lemma}
\begin{proof}
  We divide the proof into two parts.
  In the first part, we prove the result when $i=T+1$. 
  Note that we can determine whether $i=T+1$ by examining $H_i^{-m}$, as $i= T + 1$ if and only if $R_i = \emptyset$.
  In part 2, we extend this result to the general case using induction.
  
  \paragraph{Part 1.}
  If $i=T + 1$, then by definition of $f'_{T+1}$,
  \begin{align*}
    \Ex{\sum_{j= i}^{T+1} f'_j \Bigg| \bT + 1 \ge i, \bH_i^{-m} = H_i^{-m}}
      &= f'_{T+1}
    = (1-4\epsSamp) \left(
      \frac{1-\epsOpt}{1+\epsOpt}\frac{\opt}{2} - \sum_{j < i} m_j \cdot \tau^{(j)}
      \right)
  \end{align*}
  We therefore focus on the case of $i < T + 1$.
  \paragraph{Part 2}
  We prove the claim by a backwards induction on $i$.
  We note that $T \le k$ because we always have at least one sample in $|S_i|$ as long as $R_i \ne \emptyset$. 
  Therefore, if $i > k + 1$, then
  $\Pr{\bT + 1 \ge i, \bH_i^{-m} = H_i^{-m}} = 0$ and there is nothing to prove. If $i=k +1$, then
  given $\Pr{\bT + 1 \ge i, \bH_i^{-m} = H_i^{-m}} > 0$
  we must have $i=T+1$ and the 
  claim follows from Part 1.
  Assume the claim holds for $i+1$, we prove it holds for $i$ as well.
  Since in this part we assumed $i\ne T + 1$ we can conclude
  \begin{align*}
  \cbr{
   \bT + 1 \ge i, \bH_i^{-m} = H_i^{-m}
  }
  = 
  \cbr{ \bT \ge i, \bH_i^{-m} = H_i^{-m}}.
  \end{align*}

  We first give a sketch of the proof. By Lemma \ref{lm:bound_countF_always},
  we can expect $m_i \in M_i^*$ with high probability.
  As long as $m_i\in M_i^*$, by Lemma \ref{lm:bound_f_one_level}, \begin{align*}
      \Ex{f'_i| \sigma_i} > (1-4\epsSamp) \cdot m_i \tau^{(i)}.
  \end{align*}
  Furthermore, we can utilize the induction hypothesis to conclude that
  \begin{align*}
      \Ex{\sum_{j \ge i+1} f'_j \Bigg|  \sigma_i} \ge  (1-4\epsSamp) \left(
      \frac{1-\epsOpt}{1+\epsOpt}\frac{\opt}{2} - \sum_{j < i +1} m_j \cdot \tau^{(j)}
      \right)
  \end{align*}
  Adding the two expressions above results in the desired bound.
  To prove the lemma's statement, we formalize the above sketch while also carefully
  considering the case of $m_i \notin M_i^*$.

  Let the shorthand $\sigma_i^{-m}$ denote
  $\bT \ge i, \bH_i^{-m}=H_i^{-m}$.
  By the law of total expectation,
  \begin{align}
    \notag
    \Ex{\sum_{j \ge i} f'_j \Bigg| \sigma_i^{-m}}
    &=
    \notag
    \Exu{m_i \sim \bM_i | \sigma_i^{-m}}{\Ex{\sum_{j \ge i} f'_j \Bigg| \sigma_i^{-m}, \bM_i=m_i}}
    \\&=
    \notag
    \Exu{m_i \sim \bM_i | \sigma_i^{-m}}{\Ex{f'_i + \sum_{j \ge i+1} f'_j \Bigg| \sigma_i^{-m}, \bM_i=m_i}}
    \\&=
    \Exu{m_i \sim \bM_i | \sigma_i^{-m}}{\Ex{f'_i\ind{m_i \in M_i^*} + f'_i\ind{m_i \notin M_i^*} + \sum_{j \ge i+1} f'_j \Bigg| \sigma_i^{-m}, \bM_i=m_i}}
    \label{eq:jul11_1845}
  \end{align}
  Note that $\cbr{\sigma_{i}^{-m}, \bM_i=m_i}$ is the same as 
  $\cbr{\bT \ge i, \bH_i=H_i}$.
  By Lemma \ref{lm:bound_f_one_level}, if $m_i \in M_i^*$,
  and $\Pr{\bM_i=m_i\mid \sigma_{i}^{-m}} > 0$,
  then 
  \begin{equation*}
    \Ex{f'_i\big| \bT \ge i, \bH_i = H_i} \ge (1-3\epsSamp) \cdot m_i \cdot \tau^{(i)}
  \end{equation*}
  Therefore,
  \begin{align*}
    \Ex{f'_i\cdot \ind{m_i \in M_i^*} \Big| \sigma_{i}^{-m}, \bM_i=m_i}\ge (1-3\epsSamp) m_i \cdot \tau^{(i)}\cdot \ind{m_i \in M_i^*}
  \end{align*}
  Furthermore, if $m_i \notin M_i^*$, 
  \begin{equation*}
    \Ex{f'_i\big| \bT \ge i, \bH_i = H_i} \ge
    -d_i \ge 
    -(1+\epsBuck) \cdot \tau^{(i)} \cdot m_i \ge
    -2k \cdot \tau^{(i)},
  \end{equation*}
  where the first inequality follows from the fact that
  $\min\{f_i, m_i\tau^{(i)} \}\ge 0$ and the last inequality follows from $m_i \le k$ and $1+\epsBuck \le 2$. 
  Therefore,
  \begin{align*}
    \Ex{f'_i\cdot \ind{m_i \notin M_i^*} \Big| \sigma_{i}^{-m}, \bM_i=m_i}\ge  -2k \cdot \tau^{(i)} \cdot \ind{m_i \notin M_i^*}
    .
  \end{align*}
  Plugging the above inequalities in \eqref{eq:jul11_1845}, we obtain
  \begin{align*}
    \Ex{\sum_{j \ge i} f'_j \Bigg| \sigma_i^{-m}}
    &\ge
    \Exu{m_i \sim \bM_i | \sigma_i^{-m}}{(1-3\epsSamp) \cdot m_i \cdot \tau^{(i)} \cdot \ind{m_i \in M_i^*} -2k \cdot \tau^{(i)} \ind{m_i \notin M_i^*} + \Ex{\sum_{j \ge i+1} f'_j \Bigg| \sigma_i^{-m}, \bM_i=m_i}
    }
  \end{align*}
  Define $B$ as
  \begin{equation*}
    B:= 
    \epsSamp \cdot m_i \cdot \tau^{(i)} \cdot \ind{m_i \in M_i^*}
    -2k  \cdot\tau^{(i)} \cdot \ind{m_i \notin M_i^*}
    -(1-4\epsSamp) \cdot m_i \cdot \tau^{(i)}  \cdot\ind{m_i \notin M_i^*}
    .
  \end{equation*}
  It follows that
  \begin{align*}
    \Ex{\sum_{j \ge i} f'_j \Bigg| \sigma_i^{-m}}
    &\ge
    \Exu{m_i \sim \bM_i | \sigma_i^{-m}}{(1 - 4\epsSamp)\cdot m_i\cdot  \tau^{(i)}+ B + \Ex{\sum_{j \ge i+1} f'_j \Bigg| \sigma_i^{-m}, \bM_i=m_i}}
    \\&\ge %
    \Exu{m_i \sim \bM_i | \sigma_i^{-m}}{(1 - 4\epsSamp)\cdot m_i\cdot  \tau^{(i)}+ 
    \Ex{\sum_{j \ge i+1} f'_j \Bigg| \sigma_i^{-m}, \bM_i=m_i}
    }
    + 
    \Exu{m_i \sim \bM_i | \sigma_i^{-m}}{B}
  \end{align*}
  For the first term, we observe that 
  \begin{align*}
    \Ex{\sum_{j \ge i+1} f'_j \Bigg| \sigma_i^{-m}, \bM_i=m_i}
    &=
    \Exu{S_i \sim \bS_i | \sigma_i^{-m}, \bM_i=m_i}{\Ex{\sum_{j \ge i+1} f'_j \Bigg| \sigma_i^{-m}, \bM_i=m_i, \bS_i=S_i}}
    \notag
    \\&\overset{(a)}{\ge}
    \Exu{S_i \sim \bS_i | \sigma_i^{-m}, \bM_i=m_i}{ (1-4\epsSamp) \left(
      \frac{1-\epsOpt}{1+\epsOpt}\frac{\opt}{2} - \sum_{j < i+1} m_j \cdot \tau^{(j)}
      \right)}
    \notag
    \\&=(1-4\epsSamp) \left(\frac{1-\epsOpt}{1+\epsOpt}\frac{\opt}{2} - \sum_{j < i+1} m_j \cdot \tau^{(j)}\right)
    \label{qyuqwerniuqwe}
  \end{align*}
  where for $(a)$, we have used the induction hypothesis, together with iterated expectation. 
  Formally, assume that in addition to 
  $\cbr{\sigma_i^{-m}, \bM_i=m_i, \bS_i=S_i}$,
  we further condition on the value of
  $\Rp_{i+1}, R_{i+1}$. 
  Then we have conditioned on
  \begin{align*}
    \cbr{\sigma_i^{-m}, \bM_i=m_i, \bS_i=S_i, \bRp_{i+1}=\Rp_{i+1}, \bR_{i+1}=R_{i+1}},
  \end{align*}
  which
  is the same
  as conditioning on $\cbr{\bT \ge i, H_{i+1}^{-m}}$, which is the same as conditioning on
  $\cbr{\bT + 1 \ge i + 1, H_{i+1}^{-m}}$
  .
  
  Therefore,
  \begin{align*}
    &\Exu{m_i \sim \bM_i | \sigma_i^{-m}}{(1 - 4\epsSamp)\cdot m_i\cdot  \tau^{(i)}+ \Ex{
      \sum_{j \ge i+1} f'_j \Bigg| \sigma_i^{-m}, \bM_i=m_i
    }}
    \\&\ge
    \Exu{m_i \sim \bM_i | \sigma_i^{-m}}{
    (1 - 4\epsSamp)\cdot m_i\cdot  \tau^{(i)}+ 
    (1-4\epsSamp) \left(
      \frac{1-\epsOpt}{1+\epsOpt}\frac{\opt}{2} - \sum_{j < i+1} m_j \cdot \tau^{(j)}
      \right)
      }
    \\&=
    \Exu{m_i \sim \bM_i | \sigma_i^{-m}}{
    (1-4\epsSamp) \left(
      \frac{1-\epsOpt}{1+\epsOpt}\frac{\opt}{2} - \sum_{j < i} m_j \cdot \tau^{(j)}
      \right)
      }
    \\&=
    (1-4\epsSamp) \left(
      \frac{1-\epsOpt}{1+\epsOpt}\frac{\opt}{2} - \sum_{j < i} m_j \cdot \tau^{(j)}
      \right)
  \end{align*}
  where the first inequality follows from \eqref{qyuqwerniuqwe},
  To finish the proof, it suffices to show that
  \begin{align*}
    &\Exu{m_i \sim \bM_i | \sigma_i^{-m}}{B}
    \ge 0
  \end{align*}
  Note however that,
  since we assumed
  $\bT \ge i$, we can assume that
  $m_i > 0$ by definition of $M_i^*$. 
  By definition of $B$,
  \begin{align*}
    B&=
    \left(\epsSamp \cdot m_i \cdot \tau^{(i)} \cdot \ind{m_i \in M_i^*}\right)
    -\left(2k  \cdot\tau^{(i)} \cdot \ind{m_i \notin M_i^*}\right)
    -\left((1-4\epsSamp) \cdot m_i \cdot \tau^{(i)}  \cdot\ind{m_i \notin M_i^*}\right)
    \\&\ge 
    \tau^{(i)}\cdot \left(
      \ind{m_i \in M_i^*} \cdot \left(
        \epsSamp
        \right)
      -
      \ind{m_i \notin M_i^*} \cdot \left(
        2k + k
        \right)
      \right) 
    \\&\ge
    \tau^{(i)}\cdot \left(
      \ind{m_i \in M_i^*} \cdot \left(
        \epsSamp
        \right)
      -
      \ind{m_i \notin M_i^*} \cdot \left(
        3k
        \right)
      \right)
  \end{align*}
  where the first inequality follows from $m_i \le k$. 
  This further implies
  \begin{align*}
    \Exu{m_i \sim \bM_i | \sigma_i^{-m}}{B \Big| \sigma_i^{-m}, \bM_i=m_i}
    &\ge
    \tau^{(i)}
    \cdot \left(
      \Pr{\bM_i \in M_i^* | \sigma_i^{-m}} \cdot \epsSamp
      - 3k \Pr{\bM_i \notin M_i^* | \sigma_i^{-m}}
      \right)
    \\&\overset{(a)}{\ge}
    \tau^{(i)}
    \paranth{
      \frac{\epsSamp}{2} - \frac{3\epsSamp}{100}
      }
      > 0
  \end{align*}
  Where $(a)$ follows from Lemma~\ref{lm:bound_countF_always}
  and the fact that
  $1-\epsSamp/100k^{10}> 1/2$. 
\end{proof}
Setting $i=1$ in the above Lemma completes the proof of Theorem \ref{thm:approx} as
\begin{align*}
    \Ex{f(G_T \backslash D)}
    \overset{(a)}{\ge} \Ex{\sum_{i=1}^{\bT+1} f'_i}
    \ge 
    (1-4\epsSamp) \cdot 
      \frac{1-\epsOpt}{1+\epsOpt}\frac{\opt}{2}
\end{align*}
where $(a)$ follows from Lemma \ref{lm:fp_lower_bound}.
\subsection{Query complexity} 
\label{sec:query}
In this section, we analyze the query complexity of our algorithm.
First, In section \ref{sec:query_level},
we analyze the query complexity of each call to $\ReconstructF{}$
as it is the main building block of our algorithm.
Next, in section \ref{sec:query_amor}, we
show how to utilize this result for bounding the expected amortized query complexity of
our algorithm.
Throughout the section, we use $V_t$ to denote all of the elements that have been inserted, but not yet deleted. 

We start by defining a quantity that will be important in our proofs.
\begin{definition}[Potential]
  For any $i \ge 1$ and any element $e \in \Rc_i$, we define the element's
  \emph{potential with respect to level $i$}, denoted by
  $P(e, i)$,
  as the single number satisfying
  $\frac{f(e | G_{i-1})}{\tau} \in [(1+\epsBuck)^{P(e, i)-1}, (1+\epsBuck)^{P(e, i)})$. 
  For elements $e\notin \Rc_i$, we define
  $P(e, i)$ to be zero.
  Additionally, we define the \emph{potential of level $i$}
  as
  $P_i := \sum_{e \in V_t} P(e, i) = \sum_{e\in \Rc_i} P(e, i)$.
  We also define $P_i$ for $i> T+1$ to be 0.
\end{definition}
We will use $\bP_i$ instead of $P_i$ when we want to emphasize the fact that
these values are random.
\begin{lemma}\label{lm:potential_properties}
  The potential $P$ satisfies the following properties for all $i \ge 1$:
  \begin{align}
      &\forall e \in \Rc_i: P(e, i)\in \left[1, \log_{1+\epsBuck}(4k)\right] 
      \\&
      |\Rc_i| \le P_i \le |\Rc_i| \cdot \log_{1+\epsBuck}(4k)
      \\&
      P_{T+1} = 0\quad\text{ and }\quad \forall i \le T: P_{i} > 0
  \end{align}
\end{lemma}
\begin{proof}
    For the first result, note that since
    $e \in \Rc_i$, given Lemma~\ref{lm:invariant_filter},
    it must have 
    $e \in \FilterF{}(\Rc_{i-1}, G_{i-1}, \tau)$, which implies
    $f(e | G_{i-1}) \ge \tau$. Therefore, $P(e, i) \ge 1$. On the other hand,
    we note that
    since $e \in \Rc_i$, it follows that
    $e \notin D$. Therefore,
    \begin{align*}
      2k \tau = \opt \ge f(\optSet)
      \ge f(e) \ge f(e | G_{i-1}).
    \end{align*}
    It follows that $\frac{f(e | G_{i-1})}{\tau} \le 2k$.
    Note however that, by definition,
    \begin{align*}
        \frac{f(e | G_{i-1})}{\tau} 
        \ge (1+\epsBuck)^{P(e, i) - 1}
        \implies 
        P(e, i) \le \log_{1+\epsBuck}(\frac{f(e | G_{i-1})}{\tau}) + 1.
    \end{align*}
    Therefore,
    \begin{align*}
        P(e, i) \le \log_{1+\epsBuck}(2k) + 1 \le
        \log_{1+\epsBuck}(2k) +  \log_{1+\epsBuck}(2)
        =  \log_{1+\epsBuck}(4k).
    \end{align*}

    The second identity follows from the first one since
    $P(e, i) = 0$ for $e\notin \Rc_i$.

    The third identity follows from the second identity and Lemma \ref{lm:final_level_stronger}. 
\end{proof}
\begin{lemma}\label{lm:potential_decrease_always}
    For any $i \in [1, T]$,
    and any $e \in V_T$,
    $P(e, i) \ge P(e, i+1)$ and therefore,
    $P_{i} \ge P_{i+1}$.
\end{lemma}
\begin{proof}
  If $e \notin \Rc_{i+1}$, the claim holds trivially cause the right hand side will equal zero. 

  Otherwise, we note that $e\in \Rc_i$ because
  $\Rc_{i+1} \subseteq \Rc_{i}$ given Lemma~\ref{lm:invariant_filter}.
  The claim now follows from the fact that
  $G_{i} \subseteq G_{i+1}$. 
\end{proof}

\subsubsection{\ReconstructF{}}
\label{sec:query_level}
In this section, we bound the query complexity of invoking $\ReconstructF{}$.
It is clear from Algorithm \ref{alg:offline} that \ReconstructF{}
works by  sampling $S_{i}$ from $R_{i}$,
forming $R_{i+1} = \FilterF{}(R_i, G_{i})$,
and repeating the process with $i \gets i+1$. 
We can therefore bound its query complexity by
bounding
\begin{enumerate}
  \item The number of times the while-loop is executed.
    We will obtain this result in Lemma \ref{lm:num_level_reconstruct}
    after bounding the decrease in potential of each level in Lemma \ref{lm:potential_decrease}
  \item The number of queries inside the while-loop. We have already done this in Lemma \ref{lm:bound_countF}.
\end{enumerate}
We use the above results to establish a bound on the total number of queries
made by \ReconstructF{} in Lemma \ref{lm:bound_num_query_recounstruct}.

As we want to state results that hold for an arbitrary level $i$,
throughout this section, we will assume that we call \ReconstructF{}
with a value of $\Rp_j$ satisfying 
$\FilterF{}(R_j, G_{j-1}, \tau) = R_j = \Rp_j$.
Given Lemma \ref{lm:invariant_filter}, this assumption
is going to be valid in our algorithm.
We note that since we are considering the data structure
right after calling \ReconstructF{}, the values
$R_i$ and $\Rp_i$ and $\Rc_i$ will be the same
for $i \ge j$ after the call.

\begin{lemma}\label{lm:potential_decrease}
  Assume we are given sets $R_{j}, G_{j-1}$
  satisfying
  $\FilterF{}(R_j, G_{j-1}, \tau) = R_j$,
  and we invoke $\ReconstructF{}(j)$
  with $(\Rp_j, D)$ set to $(R_j, \emptyset)$,
  obtaining the (random) values
  $\bT, \bS_{j}, \dots \bS_{\bT}, \bR_{j + 1}, \dots \bR_{\bT+1}$.
  If $\epsSamp < \frac{1}{4}$,
  for each $i\ge j$,
  \begin{equation*}
    \Ex{P_{i} - \bP_{i+1} \big| \bT \ge i, \bH_i^{-m}=H_i^{-m}} \ge \frac{\epsSamp}{8} \cdot |R_i^{(b(i))}|
  \end{equation*}
  for all $H_i^{-m}$ such that $\Pr{\bT \ge i, \bH_i^{-m}=H_i^{-m}} > 0$,
  where $H_i^{-m}$ is defined as in Definition \ref{def:pre_history}.
\end{lemma}
\begin{proof}
  We first observe that since we are considering
  these values right after we invoke $\ReconstructF{}(j)$, we have $\Rc_i = R_i$. 
  We first give a sketch of the proof.
  For each $i\ge j$,
  by Lemma \ref{lm:bound_countF},
  $\bM_i \in M_i^*$
  with probability at least $1-\epsSamp/k^{10}$.
  By definition of $M_i^*$, $\bM_i \in M_i^*$ means that
  in expectation,
  at least
  $\epsSamp/4$ fraction of the elements in $R_i^{b(i)}$ satisfy $P_{i+1}(b, e) \le P_i(b, e) - 1$.
  If $\bM_i \notin M_i^*$, 
  we still have
  $P_{i+1} \le P_{i}$. This implies the bound
  as $\bM_i \in M_i^*$ holds with high probability.

  Formally,
  using the shorthand $\sigma_i^{-m}$ instead of $\bT \ge i, \bH_i^{-m}=H_i^{-m}$,
  \begin{align*}
    &\Ex{P_{i} - \bP_{i+1} \big| \sigma_i^{-m}}
    \\&=
    \Pr{\bM_i \in M_i^* | \sigma_i^{-m}} \cdot \Ex{P_{i} - \bP_{i+1} \big| \sigma_i^{-m}, \bM_i \in M_i^*}
    +
    \Pr{\bM_i \notin M_i^* | \sigma_i^{-m}} \cdot \Ex{P_{i} - \bP_{i+1} \big| \sigma_i^{-m}, \bM_i \notin M_i^*}
    \\&\ge
    \Pr{\bM_i \in M_i^* | \sigma_i^{-m}} \cdot \Ex{P_{i} - \bP_{i+1} \big| \sigma_i^{-m}, \bM_i \in M_i^*}
  \end{align*}
  where the inequality follows from the fact that $P_i - P_{i+1}$ is always non-negative given Lemma~\ref{lm:potential_decrease_always}.
  By Lemma \ref{lm:bound_countF},
  we can further bound this as
  \begin{align}
    \notag
    \Ex{P_{i} - \bP_{i+1} \big| \sigma_i^{-m}}
    &\ge
    \notag
    \Pr{\bM_i \in M_i^* | \sigma_i^{-m}} \cdot \Ex{P_{i} - \bP_{i+1} \big| \sigma_i^{-m}, \bM_i \in M_i^*}
    \notag
    \\&\ge
    \notag
    (1-\frac{\epsSamp}{100k^{10}}) \cdot \Ex{P_{i} - \bP_{i+1} \big| \sigma_i^{-m}, \bM_i \in M_i^*}
    \\&\overset{(a)}{\ge}
    \notag
    \frac{1}{2} \cdot \Ex{P_{i} - \bP_{i+1} \big| \sigma_i^{-m}, \bM_i \in M_i^*}
    \\&=
    \frac{1}{2} \cdot \Exu{m_i \sim \bM_i | \sigma_i^{-m}, \bM_i \in M_i^*}{\Ex{P_{i} - \bP_{i+1} \big| \sigma_i^{-m}, \bM_i = m_i}}
    \label{yq8wer}
  \end{align}
  where for $(a)$, we have used the fact that $\epsSamp < \frac{1}{4}$.
  Finally, we observe that
  \begin{align*}
    P_i - P_{i+1} 
    =
    \sum_{e\in R_i} P(e, i) - \sum_{e\in R_{i+1}} P(e, i+1)
    &=
    \sum_{e\in R_i} P(e, i) - \sum_{e\in R_{i}} P(e, i+1)
    \\&=
    \sum_{e \in R_i^{(b(i))}}
     \rbr{P(e, i) - P(e, i+1)}
     + \sum_{e \in R_i \backslash R_i^{(b(i))}}
     \rbr{P(e, i) - P(e, i+1)}
    \\&\overset{(a)}{\ge}
    \sum_{e \in R_i^{(b(i))}}
     \rbr{P(e, i) - P(e, i+1)}
     \\&\ge
    \sum_{e \in R_i^{(b(i))}}
    \ind{P(e, i) - P(e, i+1) \ge 1}
     \\&\overset{(b)}{{=}}
    \sum_{e \in R_i^{(b(i))}}
    \ind{f(e | G_{i}) < \tau^{(i)} \text{ or } |G_i|\ge k}
   \\&\overset{(c)}{=}
     \abs{R_i^{(b(i))}} - 
    \abs{\FilterF{}(R_i^{(b(i))}, G_i, \tau^{(i)}) }
    \\&=
     \abs{R_i^{(b(i))}} - 
    \abs{\FilterF{}(R_i^{(b(i))}, G_{i-1} \cup S_i, \tau^{(i)}) }.
  \end{align*}
  In the above derivation,  $(a)$ follows from Lemma~\ref{lm:potential_decrease_always},
   $(b)$ follows from the fact that if $|G_i| \ge k$, then $P(e, i + 1)<P(e, i)$ will hold for $e \in R_i$ because 
  $P(e, i) \ge 1$ and $P(e, i+1)=0$, and if $|G_i| \ne k$, then $P(e, i) < P(e, i+1)$ is equivalent to $f(e | G_i) < \tau^{(i)}$ for $e \in R_i^{(b(i))}$. 
  Finally, $(c)$ follows from the definition of $\FilterF{}$. 
  
  Therefore, 
  \begin{align*}
    \Ex{P_{i} - \bP_{i+1} \big| \sigma_i^{-m}}
    &\ge
    \frac{1}{2} \cdot \Exu{m_i \sim \bM_i | \sigma_i^{-m}, \bM_i \in M_i^*}{\Ex{P_{i} - \bP_{i+1} \big| \sigma_i^{-m}, \bM_i = m_i}}
    \\&\ge
    \frac{1}{2} \cdot \frac{\epsSamp}{4} |R_i^{b(i)}|
  \end{align*}
  where the first inequality follows from \eqref{yq8wer} and
  the final inequality follows from \eqref{eq:sample_count_filter}.
\end{proof}

We now prove the following auxiliary lemma
\begin{lemma}\label{lm:auxil}
    Let $\bX_0, \bX_1, \dots, \bX_n$ be a sequence of integer positive variables such that $\bX_{i} \le \bX_{i-1}$ and 
    \begin{align*}
        \Ex{\bX_i \mid \bX_1 = X_1, \dots, \bX_{i-1} = X_{i-1}}
        \le 
        (1-\eps') {X_{i-1}}. 
    \end{align*}
    Let $T$ denote the first index $i$ such that $X_T = 0$ and assume that $\bX_0=N$ for some fixed integer $N$. Then
    $\Ex{\bT} \le \frac{\log(N) + 1}{\poly(\eps')}$.
\end{lemma}
\begin{proof}
   Let $T_i$ denote the first index such that $X_{T_i} \le (1-\eps/2){^i}N$. 
   We claim that
  \begin{align}
      \Ex{\bT_i - \bT_{i-1}} \le \poly(\frac{1}{\eps'}). 
      \label{jashduqjha1822}
  \end{align}
  By Markov's inequality,
  \begin{align*}
        \Pr{{\bX_i} > (1-\eps'/2) X_{i-1} \mid \bX_1 = X_1, \dots, \bX_{i-1} = X_{i-1}}
        \le 
        \frac{1-\eps'}{1-\eps'/2} \le 1-\eps'/2. 
   \end{align*}
   Therefore, because the mean of a geometric random variable with parameter $p$ is $\frac{1}{p}$, it follows that in expectation, we need to increase the index $i$ at most
   $O(\frac{1}{\eps'})$ times before we decrease $X_i$ by a factor of $1-\eps'/2$. Since $X_i \le X_{i-1}$, this implies \eqref{jashduqjha1822}. Therefore, letting $T'_i$ denote the first index such that $X_{T'_i}\le \frac{1}{2^i} N$, it follows that 
   \begin{align*}
     \Ex{\bT'_i - \bT'_{i-1}} \le 
     \log_{\frac{1}{1-\eps'/2}}(2)
     O(\frac{1}{\eps'})
     = \frac{1}{\poly(\eps')}
   \end{align*}
   Note however that $T'_{\log(N) + 1}= T$ because
   $X_{i}$ are integers. Therefore, 
   \begin{align*}
       \Ex{\bT} \le \Ex{\bT'_{\log(N) + 1}} =
       \sum_{i=1}^{\log(N) + 1}\Ex{\bT'_{i} - \bT'_{i-1}}
       \le \frac{\log(N) + 1}{\poly(\eps')}.
   \end{align*}
\end{proof}

\begin{lemma}\label{lm:num_level_reconstruct}
  Assume we are given sets $R_{j}, G_{j-1}$
  satisfying
  $\FilterF{}(R_j, G_{j-1}, \tau) = R_j$,
  and we invoke $\ReconstructF{}(j)$
  with $(\Rp_j, D)$ set to $(R_j, \emptyset)$, obtaining the (random) values
  $\bT, \bS_{j}, \dots \bS_{T}, \bR_{j + 1}, \dots \bR_{T+1}$.
  The expected value of $\bT - j$ is bounded by
  $\poly(\log(|R_j|), \log(k), \frac{1}{\eps})$.
\end{lemma}
\begin{proof}

  As before, we note that
  $\Rc_i = R_i$ for any $i \ge j$ because we are considering the values right after invoking
  $\ReconstructF{}(j)$.
  
  Recall that by Lemma~\ref{lm:potential_properties}, we have
  \begin{align*}
    |R_i| \le P_i \le |R_i| \cdot \log_{1+\epsBuck}(4k)
  \end{align*}
  Given Lemma \ref{lm:potential_decrease},
  for $i\ge j$,
  \begin{align*}
    \Ex{P_{i} - \bP_{i+1} | \sigma_i^{-m}}
    \ge \frac{\epsSamp}{8} |R_i^{b(i)}|
    \overset{(a)}{\ge} \frac{\epsSamp}{8} \frac{|R_i|}{\ceil{\log_{1+\epsBuck}(2k)}} 
    &\ge \frac{\epsSamp}{8} \frac{P_i}{\log_{1+\epsBuck}^2(4k)} 
    \\&= \frac{\epsSamp}{8} \frac{P_i}{\log^2(4k)} \cdot \log^2(1+\epsBuck)
    \\&\overset{(b)}{\ge} \frac{\epsSamp}{8} \frac{P_i}{\log^2(4k)} \cdot \frac{\epsBuck^2}{16}
    \\&=
    \frac{\epsSamp \cdot \epsBuck^2}{160\log^2(4k)} \cdot P_i
  \end{align*}
  where for $(a)$, we have used the fact that $R_i^{b(i)}$ was
  the largest bucket,
  and for $(b)$, we have used the inequality
  $\log(1+x) \ge \frac{x}{4}$ for $x < 1$

  Setting $\eps'=\frac{\epsSamp \cdot \epsBuck^2}{160\log^2(4k)}$,
  it follows that
  \begin{align*}
      \Ex{\bP_{i+1} | \sigma_{i}^{-m} 
      }
      \le 
      (1-\eps')\cdot P_{i},
  \end{align*}
  Since the value of $\bP_1, \dots, \bP_i$ is deterministic conditioned on $\sigma_{i}^{-m}$, which further implies
  \begin{align}
      \Ex{\bP_{i+1} | 
      \bP_{i}=P_i,
      \dots, \bP_1 = P_1
      }
      \le 
      (1-\eps')\cdot P_{i}.
      \label{qjhwq78}
  \end{align}
  Formally,
  if $P_{i} \ne 0$, 
  we have $\bT \ge i$ given Lemma~\ref{lm:potential_properties}. Therefore,
  by iterated expectation,
  \begin{align*}
      \Ex{\bP_{i+1} | 
      \bP_{i}=P_i,
      \dots, \bP_1 = P_1
      }
      &=
      \Ex{\bP_{i+1} | 
      \bP_{i}=P_i,
      \dots, \bP_1 = P_1, 
      \bT \ge i
      }
      \\&=
      \Ex{
      \Ex{\bP_{i+1} | 
      \bP_{i}=P_i,
      \dots, \bP_1 = P_1, 
      \bT \ge i,
      \bH_i = H_i
      }
      }
      \\&=
      \Ex{
      \Ex{\bP_{i+1} | 
      \bT \ge i,
      \bH_i = H_i
      }
      }
      \\&\le (1-\eps')P_i
  \end{align*}
  as claimed. If $P_i = 0$, then
  \eqref{qjhwq78} holds trivially because
  $P_{i+1}\le P_{i} = 0$. 
  
  Observe that
  \begin{align*}
      \frac{1}{\eps'} = \poly(\log(k), \frac{1}{\eps})
  \end{align*}
  and that
  \begin{align*}
    P_j \le \log_{1+\epsBuck}(4k) |R_j|.
  \end{align*}
  The claim now follows from Lemma~\ref{lm:auxil}.
\end{proof}

\begin{lemma}\label{lm:bound_num_query_recounstruct}
  The expected number of queries made by calling \ReconstructF{}(i) is
  $|\Rc_i| \cdot \poly(\log(|\Rc_i|), \log(k), \frac{1}{\eps}) $, where $|\Rc_i|$ refers to the size of $\Rc_i$ after the update.
\end{lemma}
\begin{proof}
  Before stating the proofs, we note that the set $\Rc_i$ itself does not change during $\ReconstructF{}(i)$. Therefore, the value $|\Rc_i|$ is deterministic conditioned on the value of the data structure before the update.
  
  We now note the algorithm for reconstruction first sets $R_i$ to the new value of $\Rc_i$, and then starts building the levels. 
  Therefore,
  by Lemma \ref{lm:bound_countF},
  when
  each iteration of the while-loop in algorithm \ref{alg:offline} is executed,
  the call to \CalcSampleCountF{}{} makes at most 
  $\mO((|\Rc_i|) \cdot \poly(\log(k), \frac{1}{\epsilon}))$ queries.
  By Lemma~\ref{lm:num_level_reconstruct}, in expectation, the while-loop is executed at most 
  $\poly(\log(|\Rc_i|), \log(k), \frac{1}{\eps})$ times. 
  Since $\FilterF{}$ and bucketing make at most $\mO(|\Rc_i|)$ queries, the claim follows from Lemma \ref{lm:num_level_reconstruct}.
\end{proof}
\subsubsection{Amortization}
\label{sec:query_amor}
In this section, we bound the expected amorortized query complexity of the insert and delete operations.
We begin by bounding the number of levels $\bT$ at an arbitrary point in the update stream.
We state this result in Lemma \ref{lm:num_level_total}.
Next, in Lemma \ref{lm:bound_num_query_amor}, we show how to
bound the expected amorotoized query comlpexity of the algorithm by charging
back each call to \ReconstructF{} by the operations that triggered it.

\begin{lemma}\label{lm:num_level_total}
  At any point in the stream,
  the expected number of levels $\Ex{\bT}$ is 
  at most
  $\poly(\log(n), \log(k), \frac{1}{\eps})$.
\end{lemma}
\begin{proof}
  We note that this lemma
  is different from Lemma \ref{lm:num_level_reconstruct} because we are claiming that the number of levels is always bounded in expectation (not just after a call to \ReconstructF{}).
  In particular, this also means that we can no longer assume $\Rc_i=R_i$. 
  The proof follows using a similar technique as Lemma \ref{lm:num_level_reconstruct}.
  We first claim that
  a variant of Lemma \ref{lm:potential_decrease}
  still holds.
  More formally, 
  we claim that
  \begin{equation}
    \Ex{P_{i} - \bP_{i+1} \big| \bT \ge i, \bH_i^{-m}=H_i^{-m}} \ge \frac{\epsSamp}{16} \cdot |R_i^{(b(i))}|
    \label{eq:jul12_1744}
  \end{equation}
  for any $i \ge 1$.
  The proof follows with the exact same logic as the proof of Lemma \ref{lm:potential_decrease}.
  Using the shorthand $\sigma_i^{-m}$ to denote
  $\bT \ge i, \bH_i^{-m}=H_i^{-m}$ (Definition \ref{def:pre_history}),
  for all $i\ge 1$,
  \begin{align}
    \notag
    \Ex{P_{i} - \bP_{i+1} \big| \sigma_i^{-m}}
    &\ge
    \notag
    \Pr{\bM_i \in M_i^* | \sigma_i^{-m}} \cdot \Ex{P_{i} - \bP_{i+1} \big| \sigma_i^{-m}, \bM_i \in M_i^*}
    \\&\overset{(a)}{\ge}
    \notag
    (1-\frac{\epsSamp}{100k^{10}})\cdot \Ex{P_{i} - \bP_{i+1} \big| \sigma_i^{-m}, \bM_i \in M_i^*}
    \\&\ge
    \notag
    \frac{1}{2}\cdot \Ex{P_{i} - \bP_{i+1} \big| \sigma_i^{-m}, \bM_i \in M_i^*}
    \\&=
    \frac{1}{2} \cdot \Exu{m_i \sim \bM_i | \sigma_i^{-m}, \bM_i \in M_i^*}{\Ex{P_{i} - \bP_{i+1} \big| \sigma_i^{-m}, \bM_i = m_i}}
    \label{lakjsdf}
  \end{align}
  where for (a) we have now used Lemma \ref{lm:bound_countF_always} instead of \ref{lm:bound_countF}.
  
  Next, we observe that
  \begin{align*}
    P_{i} - P_{i+1}
    &=
    \sum_{e \in \Rc_i} (P(e, i) - P(e, i+1))
    \\&\ge
    \sum_{e \in R_i^{(b(i))}\cap \Rc_i}
    \ind{P(e, i) - P(e, i+1) \ge 1}
     \\&\ge
    \sum_{e \in R_i^{(b(i))}\cap \Rc_i}
    \ind{f(e | G_{i}) < \tau^{(i)} \text{ or } |G_i|\ge k}
    \\&=
    \sum_{e \in R_i^{(b(i))}}
    \ind{f(e | G_{i}) < \tau^{(i)} \text{ or } |G_i|\ge k}
    - 
    \sum_{e \in R_i^{(b(i))} \backslash \Rc_i}
    \ind{f(e | G_{i}) < \tau^{(i)} \text{ or } |G_i|\ge k}
    \\&\overset{(a)}{\ge}
    \sum_{e \in R_i^{(b(i))}}
    \ind{f(e | G_{i}) < \tau^{(i)} \text{ or } |G_i|\ge k}
    - 
    \sum_{e \in R_i^{(b(i))} \cap D}
    \ind{f(e | G_{i}) < \tau^{(i)} \text{ or } |G_i|\ge k}
    \\&\ge 
    \sum_{e \in R_i^{(b(i))}}
    \ind{f(e | G_{i}) < \tau^{(i)} \text{ or } |G_i|\ge k}
    - 
    |R_i^{(b(i))} \cap D|
   \\&=
     \abs{R_i^{(b(i))}} - 
    \abs{\FilterF{}(R_i^{(b(i))}, G_i , \tau^{(i)}) }
    - |R_i^{(b(i))} \cap D|
    \\&\overset{(b)}{\ge} 
    \abs{R_i^{(b(i))}} - 
    \abs{\FilterF{}(R_i^{(b(i))}, G_i, \tau^{(i)}) }
    - \epsDel |R_i^{(b(i))}|
  \end{align*}
  where for $(a)$, we have used the fact that
  $R_i \backslash \Rc_i \subseteq D$, and for $(b)$, we have used Lemma \ref{lm:reconstruction_condition}.
  Therefore, 
  \begin{align*}
    \Ex{P_i - \bP_{i+1} \big| \sigma_i^{-m}}
     + \frac{\epsDel}{2} |R_i^{(b(i))}|
    &\ge
    \frac{1}{2}\Exu{m_i \sim \bM_i | \sigma_i^{-m}, \bM_i \in M_i^*}{\Ex{P_i - \bP_{i+1} \big| 
    \sigma_i^{-m}, \bM_i = m_i}}
    + \frac{\epsDel}{2} |R_i^{(b(i))}|
    \\&\ge
    \frac{1}{2}\Exu{m_i \sim \bM_i | \sigma_i^{-m}, \bM_i \in M_i^*}{
      \abs{R_i^{(b(i))}} - 
    \abs{\FilterF{}(R_i^{(b(i))}, G_i, \tau^{(i)}) 
    }
    \;
    \Bigg|
    \;
    \sigma^{-m}_i, \bM_i = m_i
    }
    \\&\ge
    \frac{1}{2} \cdot \frac{\epsSamp}{4} |R_i^{b(i)}|
  \end{align*}
  where the first inequality follows from
  \eqref{lakjsdf},
  and
  the final inequality follows from Lemma \ref{lm:uniform_lazy}, together with Equation \eqref{eq:sample_count_filter}.

  Since
  $\epsDel \le \epsSamp/16$, it follows that
  \begin{align*}
   \Ex{P_i - \bP_{i+1} \big| \sigma_i^{-m}}
   \ge 
   \frac{\epsSamp}{16}
   |R_i^{b(i)}|.
  \end{align*}

  We can conclude from Lemma~\ref{lm:reconstruction_condition} that
  $|\Rc_i| \le |\Rp_i| \le \frac{3}{2}\abs{R_i} \le 2\abs{R_i}$. Therefore,
  \begin{align*}
    \Ex{P_i - \bP_{i+1} | \sigma_i^{-m}}
    \ge \frac{\epsSamp}{16} |R_i^{b(i)}|
    \ge \frac{\epsSamp}{16} \frac{|R_i|}{\log_{1+\epsBuck}(4k)} 
    &\overset{(a)}{\ge} \frac{\epsSamp}{32} \frac{|\Rc_i|}{\log_{1+\epsBuck}(4k)} 
    \\&\overset{(b)}{\ge} \frac{\epsSamp}{32} \frac{P_i}{\log_{1+\epsBuck}^2(4k)} 
    \\&= \frac{\epsSamp}{32} \frac{P_i}{\log^2(4k)} \cdot \log^2(1+\epsBuck)
    \\&\overset{(c)}{\ge} \frac{\epsSamp}{32} \frac{P_i}{\log^2(4k)} \cdot \frac{\epsBuck^2}{20}
  \end{align*}
  where for $(a)$, we have used $|\Rc_i| \le 2\abs{R_i}$,
  for $(b)$ we have used Lemma~\ref{lm:potential_properties},
  and for $(c)$ we have used the inequality 
  $\log(1+x) \ge \frac{x}{4}$ for $x < 1$
  and the assumption $\epsBuck \le 1/10$

  Defining
  $\eps' := \frac{\epsSamp}{32\log^2(4k)} \cdot 
  \frac{\epsBuck^2}{20}
  $, it follows that
  \begin{align*}
      \Ex{\bP_{i+1} | \sigma_{i}^{-m} 
      }
      \le 
      (1-\eps')\cdot P_{i},
  \end{align*}
  As before, this further implies
  \begin{align*}
      \Ex{\bP_{i+1} | 
      \bP_{i}=P_i,
      \dots, \bP_1 = P_1
      }
      \le 
      (1-\eps')\cdot P_{i}.
  \end{align*}
  Formally,
  if $P_{i} \ne 0$, 
  we have $\bT \ge i$ given Lemma~\ref{lm:potential_properties}. Therefore,
  by iterated expectation,
  \begin{align*}
      \Ex{\bP_{i+1} | 
      \bP_{i}=P_i,
      \dots, \bP_1 = P_1
      }
      &=
      \Ex{\bP_{i+1} | 
      \bP_{i}=P_i,
      \dots, \bP_1 = P_1, 
      \bT \ge i
      }
      \\&=
      \Ex{
      \Ex{\bP_{i+1} | 
      \bP_{i}=P_i,
      \dots, \bP_1 = P_1, 
      \bT \ge i,
      \bH_i = H_i
      }
      }
      \\&=
      \Ex{
      \Ex{\bP_{i+1} | 
      \bT \ge i,
      \bH_i = H_i
      }
      }
      \\&\le (1-\eps')P_i
  \end{align*}
  as claimed. If $P_i = 0$, then
  \eqref{qjhwq78} holds trivially because
  $P_{i+1}\le P_{i} = 0$. 
  Note however that
  \begin{align*}
      \frac{1}{\eps'} = \poly(\log(k), \frac{1}{\epsilon})
  \end{align*}
  and
  \begin{align*}
    P_1 \le \log_{1+\epsBuck}(4k) |\Rc_1|.
  \end{align*}
  The claim now follows from Lemma~\ref{lm:auxil}.

\end{proof}
\begin{lemma}\label{lm:bound_num_query_amor}
  The amortized query complexity of our algorithm is at most
  \begin{math}
    \poly(\log(n), \log(k), \frac{1}{\epsilon})
  \end{math}
\end{lemma}
\begin{proof}

  For queries that were not caused by \ReconstructF{}, each insertion or deletion
  can cause at most $\mO(\bT)$ queries where $\bT$ denotes the number of levels at the time of update,
  which is bounded by $\poly(\log(n), \log(k), \frac{1}{\epsilon})$ given Lemma~\ref{lm:num_level_total}
  and the fact that $|\Rc_1| \le n$.

  We will therefore consider the number of queries made by \ReconstructF{}. 
  The main idea is as follows.
  We charge the cost of each invocation of $\ReconstructF{}(i)$ to the updates that caused it, i.e., $\Rp_{i} \backslash R_i$ if $\ReconstructF{}(i)$ was triggered by an insertion and
  $R_i^{(b(i))} \cap D$ if it was triggered by a deletion. 
  Each time $\ReconstructF{}(i)$ is invoked for some $i$, the expected number of queries made is
  $|\Rc_i| \poly(\log(|\Rc_i|), \log(k), \frac{1}{\epsilon})$. 
  However, this cost is spread across at least $\frac{|\Rc_i|}{\poly(\log(k), \frac{1}{\eps})}$ updates because of the reconstruction condition. 
  Therefore, each update is charged at most $\poly(\log(k), \frac{1}{\eps})$ for each of the levels it affects. 
  Since the expected number of levels at the time of the update is at most $\poly(\log(k), \frac{1}{\eps})$ by Lemma~\ref{lm:num_level_total},
  the claim follows.
  Note that here we are crucially using the fact that the randomness in the number of queries is independent of the randomness in $\bT$.
  This is because
  $\bT$ is determined at the time of the update $u$, while the randomness in the number of queries
  is determined independent of the past at a later time,
  i.e., when \ReconstructF{} is actually triggered,
  and it is independent of the past since it comes from a random oracle in our algorithm.

  Formally, for an update $u$, a level $i$,
  and a time $t$,\footnote{Here ``time'' refers to an update in the stream. We use the word time to avoid confusion with $u$, and because the actual nature of the update at time $t$ will not play an important role.} let $Q_{u, i, t}$ denote the number of queries charged to $u$ by level $i$ for an invocation of $\ReconstructF{}(i)$ at time $t$. If $\ReconstructF{}(i)$ was not invoked, or $\ReconstructF{}(i)$ does not charge $u$, then set $Q_{u, i, t} = 0$. We need to show that
  for each $u$, 
  \begin{align}
      \Ex{\sum_{i, t} \bQ_{u, i, t}} \le \poly(\log(n), \log(k), \frac{1}{\epsilon}). 
      \label{eq:jashduqjha1625}
  \end{align}
  Let $T^u$ denote the number of levels when $u$ is inserted,
  we will show that
  \begin{align}
      \Ex{\sum_{i, t} \bQ_{u, i, t} \mid \bT^u = T^u} \le (T^{u} + 1)\poly(\log(n), \log(k), \frac{1}{\epsilon}). 
      \label{eq:jashduqjha1655}
  \end{align}
  This would imply \eqref{eq:jashduqjha1625} given Lemma~\ref{lm:num_level_total}. 
  To prove this, we first observe that
  \begin{align*}
      \sum_{i, t}\ind{Q_{u, i, t} > 0} \le T^{u} + 1.
  \end{align*}
  This is because for each update $u$
  and each level $i\le T^u + 1$,
  the level will only charge the update if $\ReconstructF{}(i)$ is invoked after the update (and before the invoking of $\ReconstructF{}(j)$ for some $j < i$), and the level will not charge the update twice. Therefore, the identity holds. 

  Note however that
  \begin{align*}
      \Ex{\bQ_{u, i, t} \mid \bT^u = T^u}
      = \Ex{\bQ_{u, i, t} \mid \bQ_{u, i, t} > 0,  \bT^u = T^u}\Pr{\bQ_{u, i, t} > 0 \mid \bT^{u} = T^u}.
  \end{align*}
  We claim that
  \begin{align}
      \Ex{\bQ_{u, i, t} \mid \bQ_{u, i, t} > 0, \bT^u = T^u} \le \poly(\log(n), \log(k), \frac{1}{\epsilon}).
      \label{eq:jashduqjha1705}
  \end{align}
  If this is shown, then it would imply 
  \eqref{eq:jashduqjha1655} because
  \begin{align*}
      &\sum_{i, t}\Ex{\bQ_{u, i, t} \mid \bT^u = T^u}
      \\&=
      \sum_{i, t} 
      \Ex{\bQ_{u, i, t} \mid \bQ_{u, i, t} > 0,  \bT^u = T^u}\Pr{\bQ_{u, i, t} > 0 \mid \bT^{u} = T^u}
      \\&{\le} 
      \rbr{
       \sum_{i, t} \Pr{\bQ_{u, i, t} > 0 \mid \bT^u = T^u}
      }
      \poly(\log(n), \log(k), \frac{1}{\epsilon})
      \\&{=}
      \rbr{
       \sum_{i, t} \Ex{\ind{\bQ_{u, i, t} > 0} \mid \bT^u = T^u}
      }
      \poly(\log(n), \log(k), \frac{1}{\epsilon})
      \\&=
      \rbr{
      \Ex{
        \sum_{i, t} \ind{\bQ_{u, i, t} > 0} \Bigg|~ \bT^u = T^u
       }
      }
      \poly(\log(n), \log(k), \frac{1}{\epsilon})
      \\&\le
      \rbr{
      \Ex{
        \bT^{u} + 1 \mid \bT^u = T^u
       }
      }
      \poly(\log(n), \log(k), \frac{1}{\epsilon})
      \\&=
      (T^u + 1)\poly(\log(n), \log(k), \frac{1}{\epsilon})
  \end{align*}

  We divide the proof of \eqref{eq:jashduqjha1705} into two parts.
  Let $\Rc_{i, t}$ denote the value of $\Rc_i$ right after time $t$.
  \begin{claim}
      \begin{align*}
          \Ex{\sum_{u'} \bQ_{u', i, t} ~\bigg|~ \bQ_{u, i, t} > 0, \bT^u = T,
          \bRc_{i, t} = \Rc_{i, t}
          } \le 
          |\Rc_{i, t}| \cdot
          \poly(\log(n), \log(k), \frac{1}{\epsilon})
      \end{align*}
      for all $i, t$ such that
      $\Pr{ \bQ_{u, i, t} > 0, \bT^u = T} >0$. 
  \end{claim}
  \begin{proof}
      \newcommand{\br}{\mathbf{r}}
      We assume that $t \ge u$ as otherwise $\bQ_{u, i, t} = 0$. 
      Let $r_t$ denote the random bits the algorithm uses in  time $t$.
      By law of iterated expectation, it suffices to show that
      \begin{align*}
          \Ex{\sum_{u'} \bQ_{u', i, t} ~\bigg|~ \bQ_{u, i, t} > 0, \bT^u = T,\bRc_{i, t} = \Rc_{i, t}, \br_1 = r_1, \dots, \br_{t-1}=r_{t-1}}
          \le 
          |\Rc_{i, t}| \cdot
          \poly(\log(n), \log(k), \frac{1}{\epsilon})
          .
      \end{align*}
      for all values of $r_1, \dots, r_{t-1}$ such that 
      \begin{align*}
          \Pr{\bQ_{u, i, t} > 0, \bT^u = T,\bRc_{i, t} = \Rc_{i, t}, \br_1 = r_1, \dots, \br_{t-1}=r_{t-1}}
          > 0
          .
      \end{align*}
      We note however that $\ind{\bQ_{u, i, t} > 0}, \bT^{u}$ are both a function of $\br_{1}, \dots, \br_{u-1}$.  
      Specifically, $\bT^{u}$ is a function of $\br_{1}, \dots, \br_{t-1}$ and $\ind{\bQ_{u, i, t} > 0}$ is a function of $\br_1, \dots, \br_{t-1}$ because $\bQ_{u, i, t} > 0$ if we invoke $\ReconstructF{}(i)$ at time $t$, which is determined by $\br_1, \dots, \br_{t-1}$, and if there have been no reconstructions of level $j \le i$ between times $u$ and $t$, which is again determined by $\br_1, \dots, \br_t$. 
      Therefore, we can drop 
      $\bQ_{u, i, t} > 0, \bT^u = T$ from the conditioning. We note that $\bRc_{i, t}$ is also deterministic conditioned on $\br_1, \dots, \br_{t-1}$, and it can also be dropped. Therefore, it suffices to show that
      \begin{align*}
          \Ex{\sum_{u'} \bQ_{u', i, t} ~\Bigg|~ \br_1 = r_1, \dots, \br_{t-1}=r_{t-1}}
          \le 
          |\Rc_{i, t}| \cdot
          \poly(\log(n), \log(k), \frac{1}{\epsilon})
          .
      \end{align*}
      Note however that $\sum_{u'} Q_{u, i, t}$ is the number of queries made at time $t$ for executing $\ReconstructF{}(i)$ (and is zero if $\ReconstructF{}(i)$ is not invoked). This is independent of $\br_1, \dots, \br_{t-1}$ and depends only on $\br_t$.
      Therefore, the identity follows from Lemma~\ref{lm:bound_num_query_recounstruct}.
  \end{proof}
  \begin{claim}
      If $|\Rc_{i, t}| > 0$,
      \begin{align*}
          Q_{u, i, t}  \le 
          \frac{\sum_{u'} Q_{u', i, t}}{|\Rc_{i, t}|
          }\poly(\log(n), \log(k), \frac{1}{\epsilon}).
      \end{align*}
  \end{claim}
  \begin{proof}
    If $\ReconstructF{}(i)$ is not called at time $t$, then the claim holds trivially. Otherwise, we need to show that the cost of this reconstruction is spread across at least
    $\frac{|\Rc_{i, t}|}{\poly(\log(k), \frac{1}{\eps})}$ updates.
    We will drop the subscript $t$ throughout the proof for convenience. 
    Let $R_i^{-}$ denote the value of $R_i$ before the invocation, and let $(R_i^{(b(i))})^{-}$ denote the largest bucket in $R_i^{-}$.
    Note that $R_i^{-}$ this is the same as the value of $R_i$ the previous time that level $i$ was constructed (i.e., $\ReconstructF{}(j)$ was invoked for some $j \le i$).
    Given the condition for invoking $\ReconstructF{}(i)$, there have been either $\ge \frac{3}{2} | R_i^{-}|$ insertions or
    $\ge |(R_i^{(b(i))})^{-}| \ge \frac{|R_i^{-}|}{\poly(\log(k), \frac{1}{\eps})}$ deletions since the last reconstruction. 
    Therefore, the cost is spread across at least $\frac{|R_i^{-}|}{\poly(\log(k), \frac{1}{\eps})}$ updates.
    Therefore,
    \begin{align*}
        Q_{u, i, t}  \le 
          \frac{\sum_{u'} Q_{u, i, t}}{|R_i^{-}|
          }\poly(\log(n), \log(k), \frac{1}{\epsilon}).
    \end{align*}
    Note however that $|\Rc_i| \le 2|R_i^{-}| + 1$ because right before the update that triggers reconstruction, the value of $|\Rc_i|$ is at most $2|R_i^{-}|$ given Lemma~\ref{lm:reconstruction_condition}. Therefore, as long as
    $|R_i^{-}| > 0$, we have $|R_i^{-}| \ge \frac{|\Rc_i|}{3}$, and the claim follows.
  \end{proof}

  Given the above claims, if $|\Rc_{i, t}| > 0$,
  \begin{align*}
         &\Ex{\bQ_{u, i, t} \mid \bQ_{u, i, t} > 0, \bT^u = T^u,
          \bRc_{i, t} = \Rc_{i, t}
          } 
          \\&\le 
          \frac{\Ex{\sum_{u'}\bQ_{u', i, t} \mid \bQ_{u, i, t} > 0, \bT^u = T^u,
          \bRc_{i, t} = \Rc_{i, t}
          } }{|\Rc_{i, t}|
          }\poly(\log(n), \log(k), \frac{1}{\epsilon}).
          \\&\le 
          \poly(\log(n), \log(k), \frac{1}{\epsilon}),
  \end{align*}
  This also holds if $|\Rc_{i, t}| = 0$ because in this case, no queries are charged to $u$. 
  Iterated expectation now
  implies \eqref{eq:jashduqjha1705}, finishing the proof.
  \end{proof}
The above lemma implies Theorem \ref{thm:query} as claimed.

\subsection{Proof of Theorem~\ref{thm:main}}
We use the parallel run algorithm outlined in Section~\ref{sec:alg_parallel}. 
We will show that this algorithm
has expected amortized query complexity of $\poly(\log(n), \log(k), \frac{1}{\eps})$ per update, and the expected submodular value of the output is at least $(\frac{1}{2} - O(\epsilon))f(\optSet)$. 
Each time an element is inserted or deleted, the operation affects at most $O(\poly(\log(k), \frac{1}{\eps})$ parallel runs. Given Theorem~\ref{thm:query}, the expected amortized query complexity of each of these algorithms per update is at most $\poly(\log(n), \log(k), \frac{1}{\eps})$. Therefore, the expected amortized query complexity of the algorithm per update is $\poly(\log(n), \log(k), \frac{1}{\eps})$ as claimed.

It remains to analyse the approximation guarantee.
Consider one of the runs with the parameter $OPT_p$ satisfying
$f(\optSet) \le OPT_p \le (1+\epsilon)f(\optSet)$. 
Consider a modified algorithm where instead of inserting elements
with $f(e) \in [\epsilon OPT_p/2k, OPT_p]$ into run $p$, we had inserted all elements $e$ such that $f(e) \le OPT_p$. We claim that this would not have changed the final output of the run $p$. This is because all elements with $f(e) \le OPT_p/2k$ are filtered out from $R_1$ and $\Rp_1$. Therefore, the values of $R_1$ and $\Rp_1$ (and therefore the values of $T, R_2, \dots, R_t, \Rp_2, \dots, \Rp_T, S_1, \dots S_T$) would not change, and so the output would stay the same. 
For the modified algorithm however, Theorem~\ref{thm:approx} implies that $f(G_T \backslash D) \ge (\frac{1}{2} - O(\epsilon))f(\optSet)$, finishing the proof.

\section{Acknowledgements}
The work is partially support by DARPA QuICC, NSF AF:Small  \#2218678, and  NSF AF:Small  \#2114269 
\bibliography{references}
\bibliographystyle{abbrv}

\end{document}